\documentclass[a4paper]{llncs}

\usepackage{microtype}

\usepackage{pgf,pgfarrows,pgfnodes,pgfautomata,pgfheaps,pgfshade}
 \usepackage{amssymb}
 \usepackage{amsfonts}
 \usepackage{amsmath}
\usepackage{verbatim}
 \usepackage{gastex}
 \usepackage{floatflt}
 \usepackage{stmaryrd}
 \usepackage{algorithm}
\usepackage{tikz}

\usepackage[]{algorithmicx}
\usepackage{algpseudocode}
\usepackage{amsmath}
\usepackage{amsmath}
\usepackage{graphicx} 
\usepackage{algcompatible}
\usepackage{algorithm}
\usepackage {array}
\usepackage{amsfonts}

\usepackage{amssymb}

\algnewcommand{\IIf}[1]{\State\algorithmicif\ #1\ \algorithmicthen}
\algnewcommand{\IElse}{\algorithmicelse }
\algnewcommand{\EndIIf}{ \unskip\ \ }

\usetikzlibrary{automata}
\usetikzlibrary{positioning}
\usetikzlibrary{arrows}
\usetikzlibrary{backgrounds}
\usetikzlibrary{snakes}
\usetikzlibrary{shapes}
\usetikzlibrary{petri}
\tikzstyle{state}=[circle,inner sep=0pt,minimum size=6mm]
\tikzstyle{stateP}=[rectangle,inner sep=0pt,minimum size=6mm]
\tikzstyle{statebis}=[diamond,thick,inner sep=0pt,minimum size=6mm]
\tikzstyle{small-state}=[circle,thick,inner sep=0pt,minimum size=5mm]
\tikzstyle{small-statebis}=[diamond,thick,inner sep=0pt,minimum size=5mm]
\tikzstyle{init-left}=[pin={[pin distance=7pt,pin edge={<-,black,thick}]left:}]
\tikzstyle{init-right}=[pin={[pin distance=7pt,pin edge={<-,black,thick}]right:}]
\tikzstyle{init-below}=[pin={[pin distance=7pt,pin edge={<-,black,thick}]below:}]
\tikzstyle{vertex}=[double,rounded corners,thick,inner sep=0pt,minimum
size=4mm]
\tikzstyle{vertexbis}=[diamond,double,thick,inner sep=0pt,minimum size=4mm]
\tikzstyle{petri-p}=[circle,thick,inner sep=0pt,minimum size=5mm]
\tikzstyle{petri-t}=[rectangle,thick,inner sep=0pt,minimum width=7mm,minimum height=1mm]
\tikzstyle{petri-t2}=[rectangle,thick,inner sep=0pt,minimum width=1mm,minimum height=7mm]
\tikzstyle{petri-tok}=[circle,inner sep=0pt,minimum size=4pt, color=black,fill=black]
\tikzstyle{gate}=[rectangle,thick,inner sep=0pt,minimum size=6mm,draw=black]

\newcommand{\nat}{\mathbb{N}}
\newcommand{\rel}{\mathbb{Z}}

\newcommand{\tuple}[1]{\langle #1 \rangle}
\newcommand{\set}[1]{\{ #1 \}}
\newcommand{\intset}[2]{[#1..#2]}

\newcommand{\trans}{\rightarrow}
\newcommand{\labeltrans}[1]{\xrightarrow{#1}}
\newcommand{\ttrans}{\Rightarrow}

\newcommand{\eventually}{\diamondsuit}
\newcommand{\always}{\Box}

\newcommand{\Interp}[1]{\llbracket #1 \rrbracket}

\newcommand{\vect}[1]{\mathtt{\textbf{#1}}}

\newcommand{\Dist}[1]{{\cal D}(#1)}

\newcommand{\Plays}{\Omega}
\newcommand{\astrat}{\sigma}
\newcommand{\Strats}{\Sigma}
\newcommand{\Outcome}[1]{\mathtt{Plays}(#1)}

\newcommand{\aevent}{{\cal A}}
\newcommand{\Prob}{\mathbb{P}}
\newcommand{\Sup}{\mbox{sup}}

\newcommand{\Enabled}[1]{\mathtt{En}(#1)}

\newcommand{\InvPre}[1]{\mathtt{InvPre}(#1)}

\newcommand{\guardmucalcul}{L^{\textit{sv}}_\mu}

\newcommand{\aplay}{\rho}
\newcommand{\last}[1]{\mathit{last}(#1)}

\newcommand{\env}{\varepsilon}
\newcommand{\Var}{\mathcal{X}}

\newcommand{\tickNO}{\times}
\newcommand{\tickOK}{\checkmark}

\title{Qualitative Analysis of VASS-Induced MDPs\thanks{Technical Report EDI-INF-RR-1422 of 
the School of Informatics at the University of Edinburgh, UK. 
(http://www.inf.ed.ac.uk/publications/report/). Full version (including proofs) of
material presented at FoSSaCS 2016 (Eindhoven, The Netherlands). arXiv.org
1512.08824 - CC BY 4.0.}}

\author{Parosh Aziz Abdulla\inst{1}\thanks{Supported by UPMARC, Uppsala Programming for Multicore Architectures Research Center.} 
\and Radu Ciobanu\inst{2} \and Richard Mayr\inst{2}\thanks{Supported by EPSRC grant EP/M027651/1.} \and
  Arnaud Sangnier\inst{3} \and Jeremy Sproston\inst{4}\thanks{Supported 
by the MIUR-PRIN project CINA and 
the EU ARTEMIS Joint Undertaking under grant agreement no.  332933 (HoliDes).}}

\institute{Uppsala University, Sweden \and University of Edinburgh, UK
  \and LIAFA, Univ Paris Diderot, Sorbonne Paris Cit\'e, CNRS, France
  \and University of Turin, Italy}

\begin{document}

\pagestyle{plain}

\maketitle

\begin{abstract}
We consider infinite-state Markov decision processes (MDPs) that are
induced by extensions of vector addition systems with states (VASS).
Verification conditions for these MDPs are described by reachability and
B\"uchi objectives w.r.t.\ given sets of control-states.
We study the decidability of some qualitative versions of these objectives, i.e.,
the decidability of whether such objectives can be achieved surely, 
almost-surely, or limit-surely.
While most such problems are undecidable in general, some are decidable for
large subclasses in which either only the controller or only the random 
environment can change the counter values (while the other side can only
change control-states).

\end{abstract}

\section{Introduction}

Markov decision processes (MDPs) \cite{Puterman:book,Filar_Vrieze:book} 
are a formal model for games on directed graphs, where certain decisions
are taken by a strategic player (a.k.a. Player 1, or controller)
while others are taken randomly (a.k.a. by nature, or the environment)
according to pre-defined probability distributions.
MDPs are thus a subclass of general 2-player stochastic games, and they are 
equivalent to 1.5-player games in the terminology of
\cite{chatterjee03simple}. They are also called ``games against nature''.

A run of the MDP consists of a sequence of visited states
and transitions on the graph. Properties of the system are expressed via
properties of the induced runs. The most basic objectives are reachability
(is a certain (set of) control-state(s) eventually visited?) and
B\"uchi objectives (is a certain (set of) control-state(s) visited
infinitely often?).

Since a strategy of Player 1 induces a probability distribution of runs
of the MDP, the objective of an MDP is defined in terms of this distribution,
e.g., if the probability of satisfying a reachability/B\"uchi objective is at
least a given constant. The special case where this constant is 1 is a key 
example of a qualitative objective. Here one asks whether Player 1 has a strategy that achieves
an objective surely (all runs satisfy the property) or 
almost-surely (the probability of the runs satisfying the property is $1$).

Most classical work on algorithms for MDPs and stochastic games has focused
on finite-state systems (e.g.,
\cite{Filar_Vrieze:book,shapley-1953-stochastic,condon-1992-ic-complexity}),
but more recently several classes of infinite-state systems have been
considered as well. For instance, MDPs and stochastic games on infinite-state probabilistic recursive systems (i.e.,
probabilistic pushdown automata with unbounded stacks) \cite{Etessami:Yannakakis:ICALP05}
and on one-counter systems \cite{BBEKW10,BBE10} have been studied. Another infinite-state probabilistic model, which is incomparable to 
recursive systems, is a suitable probabilistic extension of Vector Addition
Systems with States (VASS; a.k.a. Petri nets),
which have a finite number of unbounded counters holding natural numbers.

\smallskip

{\bf\noindent Our contribution.}
We study the decidability of probability-1 qualitative reachability and B\"uchi objectives 
for infinite-state MDPs that are induced by suitable probabilistic extensions
of VASS that we call VASS-MDPs.
(Most quantitative objectives in probabilistic VASS are either undecidable, or the solution
is at least not effectively expressible in $(\mathbb{R},+,*,\le)$
\cite{abdulla-decisive-07}.)
It is easy to show that, for general VASS-MDPs, even the simplest of these
problems, (almost) sure reachability, is undecidable.
Thus we consider two monotone subclasses: 1-VASS-MDPs and P-VASS-MDPs.
In 1-VASS-MDPs, only Player 1 can modify counter values while the
probabilistic player
can only change control-states, whereas for P-VASS-MDPs it is vice-versa.
Still these two models induce infinite-state MDPs.
Unlike for finite-state MDPs, it is possible that the value of the MDP,
in the game theoretic sense, is 1, even though there is no single strategy
that achieves value 1. For example, there can exist a family of strategies
$\astrat_\epsilon$ for every $\epsilon >0$, where playing $\astrat_\epsilon$
ensures a probability $\ge 1-\epsilon$ of reaching a given target state,
but no strategy ensures probability $1$.
In this case, one says that the reachability property holds limit-surely, but
not almost-surely (i.e., unlike in finite-state MDPs, almost-surely and 
limit-surely do not coincide in infinite-state MDPs).

We show that even for P-VASS-MDPs, all sure/almost-sure/limit-sure 
reachability/B\"uchi problems are still undecidable. 
However, in the deadlock-free subclass of P-VASS-MDPs, 
the sure reachability/B\"uchi problems become decidable (while the other
problems remain undecidable). In contrast, for 1-VASS-MDPs, the sure/almost-sure/limit-sure 
reachability problem and the sure/almost-sure B\"uchi
problem are decidable.

Our decidability results rely on two different techniques. For the
sure and almost sure problems, we prove that we can reduce them to the
model-checking problem over VASS of a restricted fragment of the
modal $\mu$-calculus that has been proved to be decidable in
\cite{abdulla-solving-13}. For the limit-sure reachability problem in
1-VASS-MDP, we use an algorithm which at each iteration reduces the
dimension of the considered VASS while preserving the limit-sure
reachability properties.

Although we do not consider the class of qualitative objectives referring to 
the probability of (repeated) reachability being strictly greater than $0$,
we observe that reachability on VASS-MDPs in such a setting
is equivalent to reachability on standard VASS (though this correspondence
does not hold for repeated reachability).

{\bf\noindent Outline.}
In Section~\ref{sec:models} we define basic notations and how VASS induce
Markov decision processes.
In Sections~\ref{sec:p-vass-mdp} and 
\ref{sec:1-VASS-MDP}
we consider verification problems for
P-VASS-MDP and 1-VASS-MDP, respectively.
In Section~\ref{sec:conclusion} we summarize the decidability results
(Table~\ref{tab:summary})
and outline future work.

\section{Models and verification problems}\label{sec:models}

Let $\nat$ (resp. $\rel$) denote the set of nonnegative integers (resp. integers). For two integers $i,j$ such that $i \leq j$ we use $\intset{i}{j}$ to represent the set $\set{k \in \rel \mid i \leq k \leq j}$. Given a set $X$ and $n \in \nat \setminus \set{0}$,  $X^n$ is the set of $n$-dimensional vectors with values in $X$.  
We use $\vect{0}$ to denote the vector such that $\vect{0}(i)=0$ for all $i \in \intset{1}{n}$. The classical order on $\rel^n$ is denoted $\leq$ and is defined by $\vect{v} \leq \vect{w}$ if and only if $\vect{v}(i) \leq \vect{w}(i)$ for all $i \in \intset{1}{n}$. We also define the operation $+$ over $n$-dimensional vectors of integers in the classical way (i.e., for $\vect{v}$, $\vect{v}' \in \rel^n$, $\vect{v} + \vect{v}'$ is defined by $(\vect{v}+\vect{v}')(i)=\vect{v}(i)+\vect{v}'(i)$ for all $i \in \intset{1}{n}$). Given a set $S$, we use $S^\ast$ (respectively $S^\omega$) to denote the set of finite (respectively infinite) sequences of elements of $S$. We now recall the notion of well-quasi-ordering (which we abbreviate as wqo). A quasi-order $(A,\preceq)$ is a wqo if for every infinite sequence of elements $a_1,a_2,\ldots$ in $A$, there exist two indices $i<j$ such that $a_i \preceq a_j$. For $n>0$, $(\nat^n,\leq)$ is a wqo. Given a set $A$ with an ordering $\preceq$ and a subset $B \subseteq A$, the set $B$ is said to be \emph{upward closed} in $A$ if $a_1 \in B$, $a_2 \in A$ and $a_1 \preceq a_2$ implies $a_2 \in B$.

\subsection{Markov decision processes}
A probability distribution on a countable set $X$ is a function $f: X \mapsto [0,1]$ such that $\sum_{x \in X}f(x)=1$. We use  $\Dist{X}$ to denote the set of all probability distributions on $X$. 
We first recall the definition of Markov decision processes.

\begin{definition}[MDPs]
A Markov decision process (MDP) $M$ is a tuple $\tuple{C,C_1,\linebreak[0] C_P,A,\trans,p}$ where: $C$ is a countable set of configurations partitioned into $C_1$ and $C_P$ (that is $C=C_1 \cup C_P$ and $C_1 \cap C_P=\emptyset$); $A$ is a set of actions; $\trans \subseteq C \times A \times C$ is a transition relation; $p: C_P \mapsto \Dist{C}$ is a partial function which assigns to some configurations in $C_P$ probability distributions on $C$ such that $p(c)(c')>0$ if and only if $c \labeltrans{a} c'$ for some $a \in A$.
\end{definition}

Note that our definition is equivalent as seeing MDPs as games played between a nondeterministic player (Player 1) and a probabilistic player (Player P). The set $C_1$ contains the nondeterministic configurations (or configurations of Player 1)  and the set $C_P$ contains the probabilistic configurations (or configurations of Player P). Given two configurations $c,c'$ in $C$, we write $c\trans c'$ whenever there exists $a \in A$ such that $c \labeltrans{a} c'$. We will  say that a configuration $c \in C$ is a \emph{deadlock} if there does not exist $c' \in C$ such that $c \trans c'$. We use $C^{df}_1$ (resp. $C^{df}_P$), to denote the configurations of Player 1 (resp. of Player P) which are not a deadlock ($df$ stands here  for deadlock free).

A \emph{play} of the MDP $M=\tuple{C,C_1,C_P,A,\trans,p}$ is either an infinite sequence of the form $c_0 \labeltrans{a_0} c_1 \labeltrans{a_1} c_2 \cdots$  or a finite sequence $c_0 \labeltrans{a_0} c_1 \labeltrans{a_1} c_2 \cdots \labeltrans{a_{k-1}} c_k$. We call the first kind of play an \emph{infinite play}, and the second one a \emph{finite play}. A play is said to be maximal whenever it is infinite or it ends in a deadlock configuration. These latter plays are called deadlocked plays. We use $\Plays$ to denote the set of maximal plays.  For a finite play $\aplay=c_0 \labeltrans{a_0} c_1 \labeltrans{a_1} c_2 \cdots \labeltrans{a_{k-1}} c_k$, let $c_k=\last{\aplay}$.  We use $\Plays^{df}_1$ to denote the set of finite plays $\aplay$ such that  $\last{\aplay} \in C^{df}_1$.

A \emph{strategy} for Player 1 is a function $\astrat : \Plays^{df}_1 \mapsto C$ such that, for all $\rho \in \Plays^{df}_1$ and $c \in C$, if $\astrat(\rho)=c$ then $\last{\aplay} \trans c$. Intuitively, given a finite play $\aplay$, which represents the history of the game so far, the strategy represents the choice of Player 1 among the different possible successor configurations from $\last{\aplay}$. We use $\Strats$ to denote the set of all strategies for Player 1. Given a strategy $\astrat \in \Strats$, an infinite play $c_0 \labeltrans{a_0} c_1 \labeltrans{a_1} c_2 \cdots$ \emph{respects} $\astrat$ if for every $k \in \nat$, we have that if $c_k \in C_1$ then $c_{k+1}=\astrat(c_0 \labeltrans{a_0} c_1 \labeltrans{a_1} c_2 \cdots c_{k})$ and if $c_k \in C_P$ then $p(c_k)(c_{k+1}) >0$. We define finite plays that respect $\astrat$ similarly. Let $\Outcome{M,c,\sigma} \subseteq \Plays$ be the set of all maximal plays of $M$ that start from $c$ and that respect $\sigma$.  


Note that once a starting configuration $c_0 \in C$ and a strategy $\sigma$ have been chosen, the MDP is reduced to an ordinary stochastic process. We define an event $\aevent \subseteq \Omega$  as a measurable set of plays and we use $\Prob(M,c,\sigma,\aevent)$ to denote the probability of event $\aevent$ starting from $c \in C$ under strategy $\astrat $. The notation $\Prob^+(M,c,\aevent)$ will be used to represent the maximal probability of event $\aevent$ starting from $c$ which is defined as $\Prob^+(M,c,\aevent)=\Sup_{\astrat \in \Strats} \Prob(M,c,\sigma,\aevent)$.

\subsection{VASS-MDPs}

Probabilistic Vector Addition Systems with States have been studied, e.g., 
in \cite{abdulla-decisive-07}.
Here we extend this model with non-deterministic choices made by a controller.
We call this new model VASS-MDPs. We first recall the definition of Vector Addition Systems with States.

\begin{definition}[Vector Addition System with States]
For $n>0$, an $n$-dimensional Vector Addition System with States (VASS) is a tuple $S=\tuple{Q,T}$ where $Q$ is a finite set of control states and  $T \subseteq Q \times \rel^n \times Q$ is the transition relation labelled with vectors of integers.
\end{definition}

In the sequel, we will not always make precise the dimension of the considered VASS. Configurations of a VASS are pairs $\tuple{q,\vect{v}} \in Q \times \nat^n$. Given a configuration $\tuple{q,\vect{v}}$ and a transition $t=\tuple{q,\vect{z},q'}$ in $T$, we will say that $t$ is \emph{enabled} at $\tuple{q'',\vect{v}}$, if $q=q''$ and $\vect{v} + \vect{z} \geq \vect{0}$. Let then $\Enabled{q,\vect{v}}$ be the set $\set{t \in T \mid t \mbox{ is enabled at } \tuple{q,\vect{v})}}$. In case the transition $t=\tuple{q,\vect{z},q'}$ is enabled at $\tuple{q,\vect{v}}$, we define $t(q,\vect{v})=\tuple{q',\vect{v}'}$ where $\vect{v}'=\vect{v} + \vect{z}$. An $n$-dimensional VASS $S$ induces a labelled transition system $\tuple{C,T,\trans}$ where $C=Q \times \nat^n$ is the set of configurations and the transition relation $\trans \subseteq C \times T \times C$ is defined as follows:
$
\tuple{q,\vect{v}} \labeltrans{t} \tuple{q',\vect{v}'} \mbox{ iff } \tuple{q',\vect{v}'}=t(q,\vect{v})
$.
%
%
 VASS are sometimes seen as programs manipulating integer variables, a.k.a. counters. When a transition of a VASS changes the $i$-th value of a vector $\vect{v}$, we will sometimes say that it modifies the value of the $i$-th counter. We show now in which manner we add probability distributions to VASS.

\begin{definition}[VASS-MDP]
A VASS-MDP is a tuple $S=\tuple{Q,Q_1,Q_P,T,\tau}$ where $\tuple{Q,T}$ is a VASS for which the set of control states $Q$ is partitioned into $Q_1$  and $Q_P$, and  $\tau : T \mapsto \nat \setminus \set{0}$ is a partial function assigning to each transition a weight which is a positive natural number.
\end{definition}

Nondeterministic (resp. probabilistic) choices are made from control states in $Q_1$ (resp. $Q_P$).
The subset of transitions from control states of $Q_1$ (resp. control states of $Q_P$) is denoted by $T_1$ (resp. $T_P$). Hence $T =T_1 \cup T_P$ with $T_1 \subseteq Q_1 \times \rel^n \times Q$ and $T_P \subseteq Q_P \times \rel^n \times Q$. A VASS-MDP $S=\tuple{Q,Q_1,Q_P,T,\tau}$ induces an MDP $M_S=\tuple{C,C_1,C_P,T,\trans,p}$ where: $\tuple{C,T,\trans}$ is the labelled transition system associated with the VASS $\tuple{Q,T}$; $C_1=Q_1 \times \nat^n$ and $C_P=Q_P \times \nat^n$; and for all $c \in C^{df}_P$ and $c' \in C$, if $c \trans c'$, the probability of going from $c$ to $c'$ is defined by 
$p(c)(c')=(\sum_{\set{t \mid t(c)=c'}} \tau(t)) / (\sum_{t \in \Enabled{c}} \tau(t))$,
whereas if $c \not \trans c'$, we have $p(c)(c')=0$. 
Note that the MDP $M_S$ is well-defined: when defining $p(c)(c')$ in the case $c \trans c'$, there exists at least one transition in $\Enabled{c}$ and consequently the sum $\sum_{t \in \Enabled{c}} \tau(t)$ is never equal to $0$. Also, we could have restricted the weights to be assigned only to transitions leaving from a control state in $Q_P$ since we do not take into account the weights assigned to the other transitions. A VASS-MDP is deadlock free if its underlying VASS is deadlock free.

%
%
Finally, as in \cite{raskin-games-05} or \cite{abdulla-solving-13}, we will see that to gain decidability it is useful to restrict the power of the nondeterministic player or of the probabilistic player by restricting their ability to modify the counters' values and hence letting them only choose a control location. This leads to the two following definitions: a \emph{P-VASS-MDP} is a VASS-MDP $\tuple{Q,Q_1,Q_P,T,\tau}$ such that for all $\tuple{q,\vect{z},q'} \in T_1$, we have $\vect{z}=\vect{0}$ and a \emph{1-VASS-MDP} is a VASS-MDP $\tuple{Q,Q_1,Q_P,T,\tau}$ such that for all $\tuple{q,\vect{z},q'} \in T_P$, we have $\vect{z}=\vect{0}$. In other words, in a P-VASS-MDP, Player 1 cannot change the counter values when taking a transition and, in a 1-VASS-MDP, it is Player P which cannot perform such an action.

\subsection{Verification problems for VASS-MDPs}

We consider qualitative verification problems for
VASS-MDPs, taking as objectives control-state reachability and repeated
reachability.
To simplify the presentation, we consider a single target 
control-state $q_F \in Q$. However, our positive decidability results
easily carry over to sets of target control-states (while the negative ones
trivially do). Note however, that asking to reach a {\em fixed target
configuration} like $\tuple{q_F,\vect{0}}$ is a very different problem
(cf. \cite{abdulla-decisive-07}).

Let $S=\tuple{Q,Q_1,Q_P,T,\tau}$ be a VASS-MDP
and $M_S$ its associated MDP. Given a control state $q_F \in Q$, we denote by $\Interp{\eventually q_F}$ the set of infinite plays $c_0 \cdot c_1
\cdot \cdots$ 
and deadlocked plays $c_0 \cdot \cdots  \cdot c_l$ of $M_S$ for which there
exists an index $k \in \nat$ such that $c_k=\tuple{q_F,\vect{v}}$ for some $\vect{v}
\in \nat^n$. Similarly, $\Interp{\always \eventually q_F}$ characterizes the set of infinite plays $c_0 \cdot c_1
\cdot \cdots$ of $M_S$ for which the set $\set{i \in \nat \mid c_i=\tuple{q_F,\vect{v}} \mbox{ for some } \vect{v} \in \nat^n}$ is infinite. Since $M_S$ is an MDP with a countable number of configurations,
we know that the sets of plays $\Interp{\eventually q_F}$ and $\Interp{\always \eventually q_F}$ are measurable (for
more details see for instance \cite{baier-principles-08}), and are hence
events for $M_S$. 
Given an initial configuration $c_0 \in Q \times \nat^n$ and a control-state $q_F \in Q$, we consider the following questions:
\begin{enumerate}
\item The \emph{sure reachability problem}: Does there exist  a strategy $\sigma \in \Sigma$ such that \\ $\Outcome{\linebreak[0]M_S,c_0,\sigma} \subseteq \Interp{\eventually q_F}$?
\item The \emph{almost-sure reachability problem}: Does there exist  a strategy $\sigma \in \Sigma$ such that $\Prob(M_S,c_0,\sigma,\Interp{\eventually q_F})=1$?
\item The \emph{limit-sure reachability problem}: Does $\Prob^+(M_S,c_0,\Interp{\eventually q_F})=1$?
\item The \emph{sure repeated reachability problem}: Does there exist  a strategy $\sigma \in \Sigma$ such that  $\Outcome{\linebreak[0]M_S,c_0,\sigma} \subseteq \Interp{\always \eventually q_F}$?
\item The \emph{almost-sure repeated reachability problem}: Does there exist  a strategy $\sigma \in \Sigma$ such that $\Prob(M_S,c_0,\sigma,\Interp{\always \eventually q_F})=1$?
\item The \emph{limit-sure repeated reachability problem}: Does $\Prob^+(M_S,c_0,\Interp{\always \eventually q_F})=1$?
\end{enumerate}
Note that sure reachability implies almost-sure reachability, which itself implies limit-sure reachability, but not
vice-versa, as shown by the counterexamples in Figure~\ref{fig:example} (see also \cite{BBEKW10}).
The same holds for repeated reachability. For the sure problems, probabilities are not taken into
account, and thus these problems can be interpreted as the answer to a two
player reachability game played on the transition system of $S$. Such games have been studied for instance in \cite{raskin-games-05,abdulla-monotonic-08,abdulla-solving-13}. Finally, VASS-MDPs subsume deadlock-free VASS-MDPs and thus decidability (resp. undecidability) results carry over to the smaller (resp. larger) class.

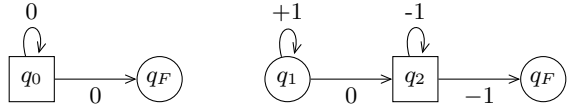
\begin{figure}[htbp]
\begin{center}
\begin{tikzpicture}[node distance=1.7cm,>=angle 45]
\node[stateP] (1) [draw=black]{$q_0$}
  edge[<-,loop above] node[auto]{$0$} (1);
\node[state] (2) [draw=black,right of=1] {$q_F$}
  edge[<-] node[auto]{$0$} (1);
\node[state] (3) [draw=black,right of=2] {$q_1$}
  edge[<-,loop above] node[auto]{$+1$} (3);
\node[stateP] (4) [draw=black,right of=3] {$q_2$}
  edge[<-] node[auto]{$0$} (3)
  edge[<-,loop above] node[auto]{-1} (4);
\node[state] (5) [draw=black,right of=4] {$q_F$}
  edge[<-] node[auto]{$-1$} (4);
\end{tikzpicture}
\end{center}

\caption{Two $1$-dimensional VASS-MDPs. The circles (resp. squares) are
  the control states of Player 1 (resp. Player P). All transitions have the same weight $1$.
From $\tuple{q_0,0}$, the state $q_F$ is reached almost-surely, but not surely, due to
the possible run with an infinite loop at $q_0$ (which has probability zero).
From $\tuple{q_1,0}$, the state $q_F$ can be reached limit-surely (by a family of
strategies that repeats the loop at $q_1$ more and more often), but not
almost-surely (or surely), since every strategy has a chance of getting
stuck at state $q_2$ with counter value zero.
}
\label{fig:example}
\end{figure}

\vspace*{-10mm}

\subsection{Undecidability in the general case}

It was shown in \cite{abdulla-monotonic-08}  that the sure reachability
 problem is undecidable for 
 (2-dimensional) 
 two player VASS. From this we can deduce that the sure reachability
 problem is undecidable for VASS-MDPs. We now present a similar proof to show the
 undecidability of the almost-sure reachability problem for VASS-MDPs. 

For all of our undecidability results we use  reductions from the undecidable
control-state reachability problem for Minsky 
machines.
A Minsky machine is a tuple $\tuple{Q,T}$ where $Q$ is a finite set of states and
$T$ is a finite set of transitions manipulating two counters, say $x_1$ and $x_2$. Each transition is a triple of the form $\tuple{q,x_i=0?,q'}$ (counter $x_i$ is tested for 0) or $\tuple{q,x_i:=x_i+1,q'}$ (counter $x_i$ is incremented) or $\tuple{q,x_i:=x_i-1,q'}$ (counter $x_i$ is decremented) where
$q,q' \in Q$. Configurations of a Minsky machine are triples in
$Q \times \nat \times \nat$. The transition relation $\ttrans$ between
configurations of the Minsky machine is then defined in the obvious way. 
Given an initial state $q_I$ and a final state $q_F$, the control-state reachability problem
asks whether there exists a sequence of configurations $\tuple{q_I,0,0} \ttrans
\tuple{q_1,v_1,v'_1} \ttrans \ldots \ttrans \tuple{q_k,v_k,v'_k}$ with $q_k=q_F$. This problem is known to be undecidable \cite{minsky-computation-67}. W.l.o.g. we assume that Minsky machines are deadlock-free and deterministic (i.e., each configuration has always a unique successor) and that the only transition leaving $q_F$ is of the form $\tuple{q_F,x_1:=x_1+1,q_F}$.

\begin{figure}[htbp]
\begin{center}
\begin{tikzpicture}[node distance=1.7cm,>=angle 45]
\node[state] (1) [draw=black]{$q_1$};
\node[state] (2) [draw=black,right of=1] {$q_2$}
  edge[<-] node[auto]{$(1,0)$} (1);
\node[state] (3) [draw=black,right of=2]{$q_3$};
\node[state] (4) [draw=black,right of=3] {$q_4$}
  edge[<-] node[auto]{$(0,-1)$} (3);
\node[state] (5) [draw=black,right of=4] {$q_5$};
\node[stateP] (5bis) [draw=black,right of=5] {}
  edge[<-] node[auto]{$(0,0)$} (5);
\node[state] (6) [draw=black,right of=5bis] {$q_6$}
  edge[<-] node[auto]{$(0,0)$} (5bis);
\node[state] (b) [draw=black,above of=5bis,yshift=-15pt] {$\bot$}
  edge[<-] node[auto]{$(-1,0)$} (5bis)
edge[<-,loop right] node[auto]{$(0,0)$} (6);
\end{tikzpicture}
\end{center}
\caption{Encoding  $\tuple{q_1,x_1:=x_1+1,q_2}$ and $\tuple{q_3,x_2:=x_2-1,q_4}$ and $\tuple{q_5,x_1=0?,q_6}$}
\label{fig:encode-minsky-VASS-MDP}
\end{figure}

\vspace{-5mm}
We now show how to reduce the control-state reachability problem to the almost-sure and limit-sure reachability problems in deadlock-free VASS-MDPs. From a Minsky machine, we construct a deadlock-free $2$-dimensional VASS-MDP for which the control states of Player 1
are exactly the control states of the Minsky machine. The encoding is presented in
Figure~\ref{fig:encode-minsky-VASS-MDP} 
where the circles (resp. squares) are the control states of Player 1 (resp. Player
P), and for each edge the corresponding weight is $1$. The state $\bot$ is an absorbing state from which the unique outgoing transition is a self loop that does not affect the values of the counters. This encoding allows us to deduce our first result.
\begin{theorem}\label{thm:VASS-MDP-undec}
The sure, almost-sure and limit-sure (repeated) reachability problems are
undecidable problems for 2-dimensional deadlock-free VASS-MDPs.
\end{theorem}

In the special case of 1-dimensional VASS-MDPs, the sure and almost-sure
reachability problems are decidable 
\cite{BBEKW10}.

\vspace{-3mm}
\subsection{Model-checking $\mu$-calculus on single-sided VASS}

It is well-known that there is a strong connection between model-checking branching time logics and games, and in our case we have in fact undecidability results for simple reachability games played on a VASS and for the model-checking of 
VASS with expressive branching-time logics \cite{esparza-decidability-94}. However for this latter point, decidability can be regained by imposing some restrictions on the VASS structure \cite{abdulla-solving-13} as we will now recall. We say that a VASS $\tuple{Q,T}$ is $(Q_1,Q_2)$-single-sided iff
$Q_1$ and $Q_2$ represents a partition of the set of states $Q$ such that for all transitions $\tuple{q,\vect{z},q'}$ in $T$ with $q \in Q_2$, we have $\vect{z}=\vect{0}$; in other words only the transitions leaving a state from $Q_1$ are allowed to change the values of the counters. In \cite{abdulla-solving-13}, it has been shown that, thanks to a reduction to games played on a single-sided VASS  with parity objectives, a large fragment of the $\mu$-calculus called $\guardmucalcul$ has a decidable model-checking problem over single-sided VASS. The idea of this fragment is that the ``always" operator $\always$ is guarded with a predicate enforcing the current control states to belong to $Q_2$. Formally, the syntax of $\guardmucalcul$ for $(Q_1,Q_2)$-single-sided VASS is given by the following grammar:
$\phi ::=  q ~\mid~ X ~\mid~ \phi \wedge \phi ~\mid~
\phi \vee \phi ~\mid~ \eventually \phi ~\mid ~ Q_2 \wedge
\always \phi ~\mid~ \mu X.\phi ~\mid~ \nu X.\phi$,
where $Q_2$ stands for the formula $\bigvee_{q \in Q_2} q$ and $X$ belongs to a set of variables $\Var$. The semantics of $\guardmucalcul$ is defined as usual: it associates to a formula $\phi$ and to an environment $\env : \Var \rightarrow 2^C$  a subset of configurations $\Interp{\phi}_\env$. We use $\env_0$ to denote the environment which assigns the empty set to any variable. Given an environment $\env$,  a variable $X \in \Var$ and a subset of configurations $C$, we use $\env[X:=C]$ to represent the environment $\env'$ which is equal to $\env$ except on the variable $X$, where we have $\env'(X)=C$. Finally the notation  $\Interp{\phi}$ corresponds to the interpretation $\Interp{\phi}_{\env_0}$. 

The problem of model-checking single-sided VASS with $\guardmucalcul$ can then be defined as follows: given a single-sided VASS $\tuple{Q,T}$, an initial configuration $c_0$ and a formula $\phi$ of $\guardmucalcul$, do we have $c_0 \in \Interp{\phi}$? 

\begin{theorem}\cite{abdulla-solving-13}\label{thm-mucalcul}
Model-checking single-sided VASS wrt. $\guardmucalcul$  is decidable.
\end{theorem}

\section{Verification of P-VASS-MDPs}
\label{sec:p-vass-mdp}

In \cite{abdulla-solving-13} it is proved that parity
games played on a single-sided deadlock-free VASS are decidable (this
entails the decidability
of model checking $\guardmucalcul$ over single-sided VASS). We will see here that in the case of P-VASS-MDPs, in which only
the probabilistic player can modify the counters, the decidability
status depends on the presence of deadlocks in the system.

\subsection{Undecidability in presence of deadlocks}

We point out that the reduction presented in Figure
\ref{fig:encode-minsky-VASS-MDP} to prove Theorem~\ref{thm:VASS-MDP-undec}
does not carry over to P-VASS-MDPs, because in that construction both players have
the ability to change the counter values. However, it is 
possible to perform a similar reduction leading to
the undecidability of verification problems for P-VASS-MDPs, the main
difference being that we crucially exploit the fact that the
P-VASS-MDP can contain deadlocks.

We now
explain the idea behind our encoding of Minsky machines into
P-VASS-MDPs. Intuitively, Player 1 chooses a transition of the Minsky
machine to simulate,
anticipating the modification of the counters values, and 
Player P is then in charge of performing the change. If Player 1
chooses a transition with a decrement and the accessed counter value is actually 0,
then Player P will be in a deadlock state and consequently the desired
control state will not be reached.  
Furthermore, if Player 1 decides to perform a zero-test when the counter
value is strictly positive, then Player P is able to punish this choice
by entering a deadlock state. Similarly to the proof of
Theorem~\ref{thm:VASS-MDP-undec}, 
Player P can test if the value of the counter is strictly greater than $0$ by
decrementing it. The encoding of the Minsky machine is presented in
Figure~\ref{fig:encode-minsky-pVASS-MDP-1}.  Note that no outgoing
edge of Player 1's states changes the counter values. Furthermore, we see that Player P reaches
the control state $\bot$ if and only if Player 1 chooses to take a
transition with a zero-test when the value of the tested counter is
not equal to $0$. Note that, with the encoding of the transition
$\tuple{q_3,x_2:=x_2-1,q_4}$, when
Player P is in the control state between $q_3$ and $q_4$, it can be in
a deadlock if the value of the second counter is not positive. In the
sequel we will see that in P-VASS-MDP without deadlocks the sure
reachability problem becomes decidable.

\begin{figure}[htbp]
\begin{center}
\begin{tikzpicture}[node distance=1.2cm,>=angle 45]
\node[state] (1) [draw=black]{$q_1$};
\node[stateP] (1bis) [draw=black,below of=1] {}
  edge[<-] node[auto]{$(0,0)$} (1);
\node[state] (2) [draw=black,below of=1bis]{$q_2$}
  edge[<-] node[auto]{$(1,0)$} (1bis);
\node[state] (3) [draw=black,right of=1,xshift=15pt]{$q_3$};
\node[stateP] (3bis) [draw=black,below of=3] {}
  edge[<-] node[auto]{$(0,0)$} (3);
\node[state] (4) [draw=black,below of=3bis]{$q_4$}
  edge[<-] node[auto]{$(0,-1)$} (3bis);
\node[state] (5) [draw=black,right of=3,xshift=15pt]{$q_5$};
\node[stateP] (5bis) [draw=black,below of=5] {}
  edge[<-] node[auto]{$(0,0)$} (5);
\node[state] (6) [draw=black,below of=5bis]{$q_6$}
  edge[<-] node[auto]{$(0,0)$} (5bis);
\node[state] (Bad) [draw=black,right of=5bis,xshift=15pt]{$\bot$}
  edge[<-] node[auto]{$(-1,0)$} (5bis)
  edge[<-,loop right] node[auto]{$(0,0)$} (5bis);
\end{tikzpicture}
\end{center}
\caption{Encoding  $\tuple{q_1,x_1:=x_1+1,q_2}$ and $\tuple{q_3,x_2:=x_2-1,q_4}$ and $\tuple{q_5,x_1=0?,q_6}$}
\label{fig:encode-minsky-pVASS-MDP-1}
\end{figure}

From this encoding we deduce the following result.

\begin{theorem}\label{thm-pvamdp-undec-sure}
The sure, almost sure and limit sure (repeated) reachability problems are
undecidable for 2-dimensional P-VASS-MDPs.
\end{theorem}

\subsection{Sure (repeated) reachability in deadlock-free P-VASS-MDPs}

Unlike in the case of general P-VASS-MDPs, we will see that the sure (repeated) reachability problem is decidable for
deadlock-free P-VASS-MDPs.
Let $S=\tuple{Q,Q_1,Q_P,T,\tau}$ be a deadlock-free P-VASS-MDP,
$M_S=(C,C_1,C_P,\trans,p)$ its associated MDP
and $q_F \in Q$ a control state. Note that because the P-VASS-MDP $S$ is
deadlock free, Player P cannot take the play to a deadlock  to avoid the control state
$q_F$, but he has to deal only with infinite plays. Since $S$ is a P-VASS-MDP, the VASS  
$\tuple{Q,T}$ is $(Q_P,Q_1)$-single-sided. 
In \cite{raskin-games-05,abdulla-monotonic-08}, it
has been shown that control-state reachability games on deadlock-free single-sided
VASS are decidable, and this result
has been extended to parity games in
\cite{abdulla-solving-13}. This implies the decidability of sure
(repeated) reachability in deadlock-free P-VASS-MDPs.
However, to obtain a generic way of verifying these systems, 
we construct a formula of $\guardmucalcul$ that characterizes the
sets of winning configurations and use then the result of Theorem~\ref{thm-mucalcul}.
Let $V^P_S$ be the set of configurations from which the answer to the
sure reachability problem (with $q_F$ as state to be reached) is
negative, i.e., $V^P_S=\set{c \in C \mid \nexists \sigma \in \Sigma
  \mbox{ s.t. } \Outcome{M_S,c,\sigma} \subseteq \Interp{\eventually
    q_F}}$ and similarly let  $W^P_S=\set{c \in C \mid
  \nexists \sigma \in \Sigma  \mbox{ s.t. } \Outcome{M_S,c,\sigma}
  \subseteq \Interp{\always \eventually q_F}}$. The next lemma relates
these two sets with a formula of $\guardmucalcul$ (where $Q_P$ corresponds to the formula $\bigvee_{q \in Q_P}$ and $Q_1$ corresponds to the formula $\bigvee_{q \in Q_1} q$). 

\begin{lemma}
\label{lem:formulae-p-pVASS-MDP}~
\begin{itemize}
\item $V^P_S=\Interp{\nu X. (\bigvee_{q \in Q \setminus
    \set{q_F}} q) \wedge (Q_1 \vee \eventually X) \wedge (Q_P \vee
  (Q_1 \wedge \always X))}$.
\item $W^P_S=\Interp{\mu Y. \nu X. \big( (\bigvee_{q \in Q \setminus
    \set{q_F}} q) \wedge (Q_1 \vee \eventually X) \wedge (Q_P \vee
   (Q_1 \wedge \always X)) \vee (q_F \wedge Q_P \wedge \eventually Y)  \vee (q_F \wedge Q_1 \wedge \always Y) \big)}$
\end{itemize}
\end{lemma}

Note that we use $(Q_P \vee
  (Q_1 \wedge \always X))$ instead of $(Q_P \vee \always X)$ so that the formulae are in
  the guarded fragment of the $\mu$-calculus. Since the two formulae belong to $\guardmucalcul$ for
the $(Q_P,Q_1)$-single-sided VASS $S$, 
decidability follows directly from Theorem \ref{thm-mucalcul}.

\begin{theorem}\label{thm-sure-reach-pvamdp}
The sure reachability and repeated reachability problem are decidable for deadlock free P-VASS-MDPs.
\end{theorem}






\subsection{Almost-sure and limit-sure reachability in deadlock-free P-VASS-MDPs}

We have seen that,
unlike for the general case, the sure reachability and sure repeated reachability
problems are decidable for deadlock free P-VASS-MDPs, with 
deadlock freeness being necessary to obtain decidability. For the corresponding almost-sure and limit-sure problems
we now show undecidability, again using a reduction from the reachability problem for two counter Minsky machines, as shown in
Figure~\ref{fig-encode-minsky-pvamdp-2}. The main difference with the construction used for the proof of Theorem~\ref{thm-pvamdp-undec-sure} lies in the
addition of a self-loop in the encoding of the transitions for decrementing a counter, 
in order to avoid deadlocks. If Player 1, from a configuration $\tuple{q_3,\vect{v}}$, 
chooses the transition $\tuple{q_3,x_2:=x_2-1,q_4}$ which decrements the second counter, then
the probabilistic state with the self-loop is entered,
and there are two possible cases:
if $\vect{v}(2) > 0$ then the probability of staying forever in this loop is $0$ and the probability of eventually going to state $q_4$ is $1$;
on the other hand, if $\vect{v}(2)=0$ then the probability of staying forever in the self-loop is $1$, since the other transition that leaves the state of Player P and which performs the decrement on the second counter effectively is not available. Note that such a construction does not hold in the case of sure reachability, because the path that stays forever in the loop is a valid path.


\begin{figure}[htbp]
\begin{center}
\begin{tikzpicture}[node distance=1.2cm,>=angle 45]
\node[state] (1) [draw=black]{$q_1$};
\node[stateP] (1bis) [draw=black,below of=1] {}
  edge[<-] node[auto]{$(0,0)$} (1);
\node[state] (2) [draw=black,below of=1bis]{$q_2$}
  edge[<-] node[auto]{$(1,0)$} (1bis);
\node[state] (3) [draw=black,right of=1,xshift=25pt]{$q_3$};
\node[stateP] (3bis) [draw=black,below of=3] {}
  edge[<-] node[auto]{$(0,0)$} (3)
 edge[<-,loop left] node[auto]{$(0,0)$} (3bis);;
\node[state] (4) [draw=black,below of=3bis]{$q_4$}
  edge[<-] node[auto]{$(0,-1)$} (3bis);
\node[state] (5) [draw=black,right of=3,xshift=15pt]{$q_5$};
\node[stateP] (5bis) [draw=black,below of=5] {}
  edge[<-] node[auto]{$(0,0)$} (5);
\node[state] (6) [draw=black,below of=5bis]{$q_6$}
  edge[<-] node[auto]{$(0,0)$} (5bis);
\node[state] (Bad) [draw=black,right of=5bis,xshift=15pt]{~$\bot$~~}
  edge[<-] node[auto]{$(-1,0)$} (5bis)
  edge[<-,loop right] node[auto]{$(0,0)$} (5bis);
\end{tikzpicture}
\end{center}
\caption{Encoding  $\tuple{q_1,x_1:=x_1+1,q_2}$ and $\tuple{q_3,x_2:=x_2-1,q_4}$ and $\tuple{q_5,x_1=0?,q_6}$}
\label{fig-encode-minsky-pvamdp-2}
\end{figure}

This allows us to deduce the following result for deadlock free P-VASS-MDPs.

\begin{theorem}\label{thm-pvamdp-undec-almsost}
The almost-sure and limit-sure (repeated) reachability problems are
undecidable for 
2-dimensional
deadlock-free P-VASS-MDPs.
\end{theorem}

\section{Verification of 1-VASS-MDPs}
\label{sec:1-VASS-MDP}

In this section, we will provide decidability results for the subclass of
1-VASS-MDPs. As for deadlock-free P-VASS-MDPs, the proofs for sure and
almost-sure problems use the decidability of $\guardmucalcul$ over
single-sided VASS, whereas the technique used to show decidability of limit-sure reachability is different.

\subsection{Sure problems in 1-VASS-MDPs}

First we show that, unlike for P-VASS-MDPs, deadlocks do not matter for
1-VASS-MDPs. The idea is that in this case, if the deadlock is in a
probabilistic configuration, it means that there is no outgoing edge (because
of the property of 1-VASS-MDPs), and hence one can add an edge to a new
absorbing state, and the same can be done for the states of Player 1. Such a
construction does not work for P-VASS-MDPs, because in that case 
deadlocks in probabilistic configurations may depend on the counter values, and
not just on the current control-state.

\begin{lemma}\label{lemma-onenpvass-nodeadlock}
  The sure (resp. almost sure, resp. limit sure) (repeated) reachability problem for 1-VASS-MDPs reduces to the
  sure (resp. almost sure, resp. limit-sure) (repeated) reachability problem for deadlock-free 1-VASS-MDPs.
\end{lemma}

Hence in the sequel we will consider only deadlock-free 1-VASS-MDPs. Let $S=\tuple{Q,Q_1,Q_P,T,\tau}$ be a deadlock-free 1-VASS-MDP. For what concerns the sure (repeated) reachability problems we can directly reuse the results from Lemma \ref{lem:formulae-p-pVASS-MDP} and then show that the complement formulae of the ones expressed in this lemma belong to  $\guardmucalcul$ for the $(Q_1,Q_P)$-single-sided VASS $\tuple{Q,T}$ (in fact the correctness of these two lemmas did not depend on the fact that we were considering P-VASS-MDPs). Theorem \ref{thm-mucalcul} allows us to retrieve the decidability results already expressed in \cite{raskin-games-05} (for sure reachability) and \cite{abdulla-solving-13} (for sure repeated reachability).

\begin{theorem}\label{thm-sure-1vamdp}
  The sure (repeated) reachability problem is decidable for 1-VASS-MDPs.
\end{theorem}

\subsection{Almost-sure problems in 1-VASS-MDPs}

We now move to the case of almost-sure problems in 1-VASS-MDPs. We consider a deadlock free 1-VASS-MDP $S=\tuple{Q,Q_1,Q_P,T,\tau}$ and its associated MDP $M_S=\tuple{C,C_1,C_P,\trans,p}$. We will see that, unlike for P-VASS-MDPs, it is here also possible to characterize by formulae of $\guardmucalcul$ the two following  sets: $V^1_{AS}=\set{c \in C \mid \exists \sigma \in \Sigma  \mbox{ such that  } \Prob(M_S,c,\sigma,\Interp{\eventually q_F})=1}$ and $W^1_{AS}=\set{c \in C \mid \exists \sigma \in \Sigma  \mbox{ such that  } \Prob(M_S,c,\sigma,\Interp{\always \eventually q_F})=1}$, i.e. the set of configurations from which Player 1 has a strategy to reach the control state $q_F$, respectively to visit infinitely often $q_F$, with probability 1.

We begin with introducing the following formula of $\guardmucalcul$ based on the variables $X$ and $Y$:
$
\InvPre{X,Y}= (Q_1 \wedge \eventually (X \wedge Y)) \vee (\eventually Y \wedge Q_P \wedge  \always X)
$.
Note that $\InvPre{X,Y}$ is a formula of $\guardmucalcul$ for the $(Q_1,Q_P)$-single-sided VASS $\tuple{Q,T}$. Intuitively, this formula represents the set of configurations from which (i)~Player 1 can make a transition to the set represented by the intersection of the sets characterized by the variables $X$ and $Y$ and (ii)~Player P can make a transition to the set $Y$ and cannot avoid making a transition to the set $X$.\\

\subsubsection{Almost sure reachability.}

We will now prove that $V^1_{AS}$ can be characterized by the following formula of $\guardmucalcul$: $\nu X. \mu Y. ( q_F \vee \InvPre{X,Y})$. Note that a similar result exists for finite-state MDPs, see e.g. \cite{chatterjee-qualitative-09}; this result in general does not extend to infinite-state MDPs, but in the case of VASS-MDPs it can be applied. Before proving this we need some intermediate results.

We denote by $E$ the set $\Interp{\nu X. \mu Y. \big ( q_F \vee \InvPre{X,Y}\big )}_{\env_0}$. Since $\nu X. \mu Y. \big ( q_F \vee \InvPre{X,Y}\big)$ is a formula of $\guardmucalcul$ interpreted over the single-sided VASS $\tuple{Q,T}$, we can show that $E$ is an upward-closed set. We now need another lemma which states that there exists $N \in \nat$ and a strategy for Player 1 such that, from any configuration of $E$, Player 1 can reach the control state $q_F$ in less than $N$ steps and Player P cannot take the play outside of $E$. The fact that we can bound the number of steps is crucial to show that $\Interp{\nu X. \mu Y. \big ( q_F \vee \InvPre{X,Y}\big )}_{\env_0}$ is equal to $V^1_{AS}$. For infinite-state MDPs where this property does not hold, our techniques do not apply.

\begin{lemma}\label{lemma-N-strat}
There exists $N \in \nat$ and a strategy $\sigma$ of Player 1 such that for all $c \in E$, there exists a play $c \cdot c_1 \cdot c_2 \cdot \ldots$ in $\Outcome{M_S,c,\sigma}$ satisfying the three following properties: (1) there exists $0 \leq i \leq N$ such that $c_i \in \Interp{q_F}$; (2) for all $0 \leq j \leq i$, $c_j \in E$; (3) for all $0 \leq j \leq i$, if $c_j \in C_P$ then for all $c'' \in C$ such that $c_j \trans c''$, we have $c'' \in E$.
\end{lemma}

This previous lemma allows us to characterize $V^1_{AS}$ with a formula of
$\guardmucalcul$. The proof of the following result uses the fact that the
number of steps is bounded, and also the fact that the sets described by closed
$\guardmucalcul$ formulae are upward-closed. This makes the fixpoint iteration
terminate in a finite number of steps.

\begin{lemma}\label{lem-probab-formula}
  $V^1_{AS}=\Interp{\nu X. \mu Y. ( q_F \vee \InvPre{X,Y})}$.
\end{lemma}

Since $\tuple{Q,T}$ is $(Q_1,Q_P)$-single-sided and since the formula associated to $V^1_{AS}$ belongs to $\guardmucalcul$, from Theorem \ref{thm-mucalcul} we deduce the following theorem.

\begin{theorem}\label{thm-almost-sure-1vamdp}
  The almost-sure reachability problem is decidable for 1-VASS-MDPs.
\end{theorem}

\subsubsection{Almost sure repeated reachability.}

For the case of almost sure repeated reachability we reuse the previously introduced formula $\InvPre{X,Y}$.  We can perform a reasoning similar to the previous ones and provide a characterization of the set $W^1_{AS}$.

\begin{lemma}\label{lem-probab-formula-repeat}
  $W^1_{AS}=\Interp{\nu X. \InvPre{X,\mu Y.( q_F \vee \InvPre{X,Y})}}$.
\end{lemma}

As previously, this allows us to deduce the decidability of the almost sure repeated reachability problem for 1-VASS-MDP.

\begin{theorem}\label{thm-almost-sure-repeat-1vamdp}
  The almost sure repeated reachability problem is decidable for 1-VASS-MDPs.
\end{theorem}

\subsection{Limit-sure reachability in 1-VASS-MDP}

We consider a slightly more general version of the limit-sure reachability
problem with a set $X \subseteq Q$ of target states instead of a single state
$q_F$, i.e., the standard case corresponds to $X = \{q_F\}$.

We extend the set of natural numbers $\mathbb{N}$ 
to $\mathbb{N}_{*}= \mathbb{N} \bigcup \{*\}$
by adding an element $*\notin\mathbb{N}$ with $*+j=*-j=*$ and $j < *$
for all $j \in \nat$. We consider then the set of vectors $\mathbb{N}_{*}^{d}$.
The projection of a vector $\vect{v}$ in $\mathbb{N}^{d}$ by eliminating components
that are indexed by a natural number $k$ is defined by
$proj_k (\vect{v})(i) = \vect{v}(i) $ if  $i\neq k$ and  $proj_k (\vect{v})(i)=*$ otherwise

Let $Q_c$ represent control-states which are indexed by a color.
The coloring functions $col_i:Q \to Q_c$ create colored copies of
control-states by $col_i(q)= q_i$.


Given a 1-VASS-MDP $S=\tuple{Q,Q_1,Q_P,T,\tau}$ of dimension $d$, an index $k \le d$ and a color $i$,
the colored projection is defined as:
\[
Proj_k(S,d,i)=\tuple{col_i(Q), col_i(Q_1), col_i(Q_P), proj_{k,i} (T),\tau_{k,i}}
\]
where $proj_{k,i} (T)= \{proj_{k,i} (t) | t \in T\}$
is the projection of the set of transitions $T$
and  $proj_{k,i} (t)= \tuple{col_i(x), proj_k(\vect{z}),col_i(y)}$
is the projection of transition $t=\tuple{x,\vect{z},y}$
obtained by removing component $k$ and 
coloring the states $x$ and $y$ with color $i$.
The transition weights carry over, i.e., 
$\tau_{k,i}(t') = \sum \{\tau(t)\,|\, proj_{k,i} (t)=t'\}$.

We define the functions $state: Q \times\mathbb{N}_{*}^{d} \to Q$ and 
${\it count}: Q \times\mathbb{N}_{*}^{d} \to \mathbb{N}_{*}^{d}$ s.t 
for a configuration $c_i = \tuple{q, \vect{v}}$, where $q  \in Q$ and
$\textbf{v} \in \mathbb{N}_{*}^{d}$  
we have that $state(q, \textbf{v})=q$ and $count(q, \textbf{v})= \textbf{v}$.
For any two configurations $c_1$ and $c_2$, we write $c_1 \prec c_2$ to denote
that $state(c_1)=state(c_2)$, 
and there exists a nonempty set of indexes $I$ where for every $i \in I$ ,
$count(c_1)(i) < count(c_2)(i)$, 
whereas for every index $j \notin I$, $0<j \le d$, $count(c_1)(j) = count(c_2)(j)$.

Algorithm~\ref{alg:limitsure} reduces the dimension of the limit-sure 
reachability problem for 1-VASS-MDP by a construction resembling the
Karp-Miller tree \cite{KaMi:schemata}. It takes as input a 1-VASS-MDP $S$ of some dimension
$d>0$ with a set of target states $X$. It outputs a new 1-VASS-MDP $S'$
of dimension $d-1$ and a new set of target states $X'$ such that 
$M_S$ can limit-surely reach $X$ iff $M_{S'}$ can limit-surely reach $X'$.
In particular, in the base case where $d-1=0$, the new system $S'$ has
dimension zero and thus induces a finite-state MDP $M_{S'}$, for which limit-sure
reachability of $X'$ coincides with almost-sure reachability of $X'$, which is 
known to be decidable in polynomial time.
Algorithm~\ref{alg:limitsure} starts by exploring all branches of 
the computation tree of $S$ (and adding them to $S'$ as the so-called 
initial {\em uncolored part}) until it encounters a configuration that is either
(1) equal to, or (2) strictly larger than a configuration encountered previously
on the same branch. In case (1) it just adds a back loop to the point
where the configuration was encountered previously.
In case (2), it adds a modified copy of $S$ (identified by a unique color)
to $S'$. This so-called colored subsystem is similar to $S$ except that those counters that have
strictly increased along the branch are removed. The intuition is that 
these counters could be pumped to arbitrarily high values and thus present no
obstacle to reaching the target.
Since the initial uncolored part is necessarily finite (by Dickson's Lemma)
and each of the finitely many colored subsystems only has 
dimension $d-1$ (since a counter is removed;
possibly a different one in different colored subsystems), the resulting
1-VASS-MDP $S'$ has dimension $d-1$. 
The set of target states $X'$ is defined as the union of all appearances of
states in $X$ in the uncolored part, plus all colored copies of states from
$X$ in the colored subsystems.

\begin{algorithm}[htbp]
\begin{algorithmic}[1]
\REQUIRE $S=\tuple{Q,Q_1,Q_P,T,\tau}$  1-VASS-MDP, dimension $d > 0$, $c_0 = \tuple{q_0,\vect{v}} \in Q \times \mathbb{N}^{d}$

 $X\subseteq Q$ - set of target states
  \ENSURE $S'=\tuple{Q',Q_1',Q_P',T',\tau'}$; $c_0' = \tuple{q_0',\vect{0}}$;
  $X' \subseteq Q'$; 
$\lambda : Q' \to ((Q  \bigcup Q_c)\times  \mathbb{N}_*^{d})$
  \STATE $Q' \leftarrow \varnothing$;  $Q_1' \leftarrow \varnothing$; $Q_P' \leftarrow \varnothing$; $T'\leftarrow \varnothing$; $\tau' \leftarrow \varnothing$; 
  \STATE $\mathrm{new}(q')$; $q_0' \leftarrow q'$; $\lambda(q')\leftarrow c_0$; $Q' \leftarrow \{q'\}$; $i \leftarrow 0$

    	\IIf{$state(\lambda(q')) \in Q_1$} $Q_1'\leftarrow\{q'\}$ \IElse $\ Q_P'\leftarrow \{q'\}$ \EndIIf
    \STATE  ToExplore $\leftarrow$ $\{q'\} $
    \WHILE{ToExplore $\neq$ $\varnothing$}
     \STATE Pick and remove a $q$ $\in$ ToExplore
     
      \IF{$\exists q'$. $q'$ is previously on the same branch as $q$ and $\lambda(q') \prec \lambda(q)$ }
       \STATE get indexes $I$ in which the counter is increasing

        \STATE pick and remove the first index $k$  from $I$ 
        
         \STATE $i  \leftarrow i +1$;  // increase color index
         

       \STATE $\mathrm{new}(q'')$;

       \STATE $\lambda(q'') \leftarrow \tuple{col_i(state(\lambda(q))), proj_k(count(\lambda(q)))} $
       
       	\IIf{$state(\lambda(q))\in Q_1$} 
       	$Q_1'\leftarrow Q_1'\bigcup \{q''\}$ \IElse $\ Q_P'\leftarrow Q_P'\bigcup \{q''\}$
       	 \EndIIf      
       \STATE $T' \leftarrow T' \bigcup \{\tuple{q, \vect{0}, q''}\}$; $\tau'(\tuple{q, \vect{0}, q''})=1$;
        \STATE  $Q_1' \leftarrow Q_1' \bigcup col_{i}(Q_1);\ $ $Q_P' \leftarrow Q_P' \bigcup col_{i}(Q_P)$; $T' \leftarrow T' \bigcup proj_{k,i}(T);\ $
           \STATE  $X'  \leftarrow X' \bigcup col_{i}(X)$; $\tau' \leftarrow \tau' \cup \tau_{k,i}$

      
       \ELSE
      \FOR{ every $t=\tuple{x,\vect{z},y}$ in $T$ such that $t \in \Enabled{\lambda(q)}$}
    \IF {$\exists q'$. $q'$ is previously on the same branch as $q$ and $t(\lambda(q))=\lambda(q')$ }
    \STATE $T' \leftarrow T' \bigcup \{\tuple{q,\vect{z},q'}\}; \ $ 
       \ELSE 
        \STATE $\mathrm{new}(q')$; $\lambda(q') \leftarrow  t(\lambda(q))$
        \STATE $T' \leftarrow T' \bigcup \{\tuple{q,\vect{z},q'}\}; \ $ $\tau'(\tuple{q,\vect{z},q'}) \leftarrow \tau(t)$

    	\IIf{$state(\lambda(q')) \in Q_1$} $Q_1'\leftarrow Q_1'\bigcup \{q'\}$ \IElse $\ Q_P'\leftarrow Q_P'\bigcup \{q'\}$ \EndIIf
     \IIf{$state(\lambda(q')) \in X$} $X'\leftarrow X'\bigcup \{q'\}$   \EndIIf
      \STATE $ToExplore \leftarrow ToExplore \bigcup \{q'\}$
    \ENDIF
    
     \ENDFOR
        \ENDIF
     \ENDWHILE
\end{algorithmic}

\caption{Reducing the dimension of the limit-sure reachability problem.}\label{alg:limitsure}
\end{algorithm}

By Dickson's Lemma, the conditions on line 7 or line 19 of the algorithm must eventually hold 
on every branch of the explored computation tree. Thus, it will terminate.

\begin{lemma}\label{lem:alg-limitsure-termination}
Algorithm~\ref{alg:limitsure} terminates.
\end{lemma}

The next lemma states the correctness of Algorithm~\ref{alg:limitsure}.
Let $S=\tuple{Q,Q_1,Q_P,T,\tau}$ be 1-VASS-MDP of dimension $d > 0$ with initial
configuration $c_0 = \tuple{q_0,\vect{v}}$ and $X\subseteq Q$ a set of target states.
Let $S'=\tuple{Q',Q_1',Q_P',T',\tau'}$ with initial configuration 
$c_0' = \tuple{q_0',\vec{0}}$ and set of target states $X' \subseteq Q'$ be the 
$(d-1)$ dimensional 1-VASS-MDP produced by
Algorithm~\ref{alg:limitsure}. As described above we have the
following relation between these two systems.

\begin{lemma}\label{lem:limitsure-correctness-algo}
$\mathbb{P}^+(M_S, c_0, \Interp{\eventually X}) = 1$ iff
$\mathbb{P}^+(M_{S'}, c_0', \Interp{\eventually X'}) = 1$.
\end{lemma}

By applying the result of the previous lemma iteratively until we obtain
a finite-state MDP, we can deduce the following theorem.

\begin{theorem}\label{thm-limitsure-reach-1VAMDP}
The limit-sure reachability problem for 1-VASS-MDP is decidable.
\end{theorem}

\section{Conclusion and Future Work}\label{sec:conclusion}

Table~\ref{tab:summary} summarizes our results on the decidability of
verification problems for subclasses of VASS-MDP. 
The exact complexity of most problems is still open. 
Algorithm~\ref{alg:limitsure} relies on
Dickson's Lemma for termination, and the algorithm deciding the model-checking problem of
Theorem~\ref{thm-mucalcul} additionally uses the Valk-Jantzen construction
repeatedly. However, all these problems are at least as hard as control-state
reachability in VASS, and thus EXPSPACE-hard \cite{esparza-decidability-94}.

The decidability of the limit-sure repeated reachability problem for
1-VASS-MDP is open. A hint of its difficulty is given by the fact that there
are instances where the property holds even though a small chance of reaching
a deadlock cannot be avoided from any reachable configuration. In particular,
a solution would require an analysis of the long-run behavior of multi-dimensional random
walks induced by probabilistic VASS. However, these may exhibit strange 
nonregular behaviors for dimensions $\ge 3$, as described in
\cite{BKKN:LICS2015} (Section 5).

\begin{table}[htb]
\begin{tabular}{l||c|c|c}
                         & P-VASS-MDP & df P-VASS-MDP & 1-VASS-MDP\\ \hline
sure reachability        & $\tickNO$ (Thm.~\ref{thm-pvamdp-undec-sure}) & $\tickOK$ (Thm.~\ref{thm-sure-reach-pvamdp}) & $\tickOK$ (Thm.~\ref{thm-sure-1vamdp}) \\ \hline
almost-sure reachability & $\tickNO$ (Thm.~\ref{thm-pvamdp-undec-sure}) & $\tickNO$ (Thm.~\ref{thm-pvamdp-undec-almsost}) & $\tickOK$ (Thm.~\ref{thm-almost-sure-1vamdp}) \\ \hline
limit-sure reachability  & $\tickNO$ (Thm.~\ref{thm-pvamdp-undec-sure}) & $\tickNO$ (Thm.~\ref{thm-pvamdp-undec-almsost}) & $\tickOK$ (Thm.~\ref{thm-limitsure-reach-1VAMDP})\\ \hline
sure repeated reachability             & $\tickNO$ (Thm.~\ref{thm-pvamdp-undec-sure}) & $\tickOK$ (Thm.~\ref{thm-sure-reach-pvamdp}) & $\tickOK$ (Thm.~\ref{thm-sure-1vamdp}) \\ \hline
almost-sure repeated reachability      & $\tickNO$ (Thm.~\ref{thm-pvamdp-undec-sure}) & $\tickNO$ (Thm.~\ref{thm-pvamdp-undec-almsost}) & $\tickOK$ (Thm.~\ref{thm-almost-sure-repeat-1vamdp}) \\ \hline
limit-sure repeated reachability       & $\tickNO$ (Thm.~\ref{thm-pvamdp-undec-sure}) & $\tickNO$ (Thm.~\ref{thm-pvamdp-undec-almsost}) & Open
\end{tabular}
\smallskip
\caption{Decidability of verification problems for P-VASS-MDP, deadlock-free
P-VASS-MDP and 1-VASS-MDP. A $\tickOK$ stands for decidable and a $\tickNO$
for undecidable.}\label{tab:summary} 
\end{table}


\newpage
\appendix
\section{Proofs of Section \ref{sec:p-vass-mdp}}

\subsection{Proof of Lemma \ref{lem:formulae-p-pVASS-MDP}}

Even if the result of this lemma is quite standard, in order to be consistent
we provide the proof, in particular to be sure that the fact that we
are dealing with infinite state systems does not harm the reasoning.

\begin{lemma}
\label{lem:sure-p-pVASS-MDP}~
$V^P_S=\Interp{\nu X. (\bigvee_{q \in Q \setminus
    \set{q_F}} q) \wedge (Q_1 \vee \eventually X) \wedge (Q_P \vee
  (Q_1 \wedge \always X))}$.
\end{lemma}

\begin{proof}
We denote by $U$ the set $\Interp{\nu X. (\bigvee_{q \in Q \setminus
    \set{q_F}} q) \wedge (Q_1 \vee \eventually X) \wedge (Q_P \vee
   (Q_1 \wedge \always X))}$. We will consider the function $g: 2^C \mapsto
2^C$ such that for each set of configurations $C' \subseteq C$, we have $g(C')=\Interp{\bigvee_{(q \in Q \setminus
    \set{q_F}} q) \wedge (Q_1 \vee \eventually X) \wedge (Q_P \vee
   (Q_1 \wedge \always X))}_{\env_0[X:=C']}$, i.e. the set $\Interp{\bigvee_{q \in Q \setminus
    \set{q_F}} q) \wedge (Q_p \wedge \eventually X) \wedge (Q_1 \wedge
   (Q_1 \wedge \always X))}$ where $X$ is interpreted as $C'$. Note
 that $U$ is then the greatest fixpoint of $g$ and hence $U=g(U)$.\\

We first prove that $U \subseteq V^P_S$. Let $c$ be a configuration in
$C$ such that $c \notin V^P_S$. Then there exists a strategy
$\sigma \in \Sigma$ such that  $\Outcome{M_S,c,\sigma} \subseteq
\Interp{\eventually q_F}$. We consider such a strategy $\sigma$ and  we reason by contradiction assuming that $c \in U$.

Let us show that
since $c \in U$, there exists an infinite play $c_0 \rightarrow c_1
\rightarrow c_2 \rightarrow \ldots$ in $\Outcome{M_S,c,\sigma}$ such
that $c_0=c$ and $c_i$ in $U$ for all $i \in \nat$. We prove in fact
by induction 
that if $c_0 \rightarrow c_1 \ldots \rightarrow c_k$ is a finite play
in $M_S$ respecting $\sigma$ with $c_0=c$ and such that $c_i \in U$ for $i \in [0..k]$,
then there exists $c_{k+1} \in U$ such that  $c_0 \rightarrow c_1
\ldots \rightarrow c_{k+1}$ is a play in $M_S$ respecting $\sigma$. The
base case is obvious since $c \in U$. We assume now that $c_0 \rightarrow c_1 \ldots \rightarrow c_k$ is a finite play in
$M_S$ respecting $\sigma$ such that $c_i \in U$ for $i \in
[0..k]$. Because there is no deadlock, if $c_k \in C_1$ then there
exists $c_{k+1} \in C$ satisfying $c_{k+1}=\sigma(c_0 \rightarrow c_1
\ldots \rightarrow c_k)$. Furthermore since $c_k \in U$ and $U=g(U)$, we deduce
that, for all $c \in C$ such that $c_k \rightarrow c$, we have $c \in
U$, and consequently $c_{k+1} \in U$. On the other hand, if $c_k \in
C_P$, then since $c_k \in U$, there exists $c_{k+1} \in
U$ such that $c_k \rightarrow c_{k+1}$. In both cases, we have that
$c_0 \rightarrow c_1 \ldots \rightarrow c_k \rightarrow c_{k+1}$
is a play in $M_S$ which respects $\sigma$. 

We deduce  the existence of  an infinite play $c_0 \rightarrow c_1
\rightarrow c_2 \rightarrow \ldots$ in $\Outcome{M_S,\linebreak[0]c,\sigma}$ such
that $c_0=c$ and $c_i$ in $U$ for all $i \in \nat$. Note that, because
$U=g(U)$, we have that $U \subseteq Q \setminus \set{q_F}$. However, because $c_0 \rightarrow c_1
\rightarrow c_2 \rightarrow \ldots$ in $\Outcome{M_S,c,\sigma}$, we
also have that $c_0 \rightarrow c_1
\rightarrow c_2 \rightarrow \ldots \in \Interp{\eventually q_F}$, which
is a contradiction. Consequently, we have $c \notin U$ and this allows
us to conclude that $U
\subseteq V^P_S$.\\

We now prove that $V^P_S \subseteq U$. By the Knaster-Tarski Theorem, since $U$ is the greatest fixpoint of
$g$, we know that $U = \bigcup \set{C'
  \subseteq C \mid C' \subseteq g(C')}$. It hence suffices to show that
$V^P_S \subseteq g(V^P_S)$. Let $c=\tuple{q,\vect{v}}$ be a
  configuration in $V^P_S$. First note that by definition of $V^P_S$,
  we have $q \neq q_F$. We reason then by a case
analysis to prove that $c \in g(V^P_S)$. First assume $c \in C_1$. Then by definition of
$V^P_S$, for all $c'=\tuple{q',\vect{v}'}$ in $C$ satisfying $c
\rightarrow c'$, we have $c' \in V^P_S$, otherwise Player 1 would have
a strategy to reach $q_F$ from $c$ which will consist in taking the
transition leading to $c'$. This allows us to deduce that $c \in
g(V^P_S)$. Assume now that $c \in C_P$. Then, because $c \in V^P_S$ and
$S$ is deadlock free, there necessarily exists $c'$ such $c
\rightarrow c'$ and $c' \in V^P_S$ (otherwise there would be a strategy to
reach $q_F$ from all states $c'$ such that $c \rightarrow c'$ and hence $c$
would not be in $V^P_S$). So also in this case we have $c \in
g(V^P_S)$. Hence we have $V^P_S \subseteq g(V^P_S)$ which allows to
deduce that $V^P_S \subseteq U$.
\qed
\end{proof}

\begin{lemma}
\label{lem:almostsure-p-pVASS-MDP}
$W^P_S=\Interp{\mu Y. \nu X. \big( (\bigvee_{q \in Q \setminus
    \set{q_F}} q) \wedge (Q_1 \vee \eventually X) \wedge (Q_P \vee
   (Q_1 \wedge \always X)) \vee (q_F \wedge Q_P \wedge \eventually Y)  \vee (q_F \wedge Q_1 \wedge \always Y) \big)}$.
\end{lemma}

\begin{proof}
We denote by $U$ the set $\Interp{\mu Y. \nu X. \big( (\bigvee_{q \in Q \setminus
    \set{q_F}} q) \wedge (Q_1 \vee \eventually X) \wedge (Q_P \vee
   (Q_1 \wedge \always X)) \vee (q_F \wedge Q_P \wedge \eventually Y)  \vee (q_F \wedge Q_1 \wedge \always Y) \big)}$ and we consider the function $h: 2^C \mapsto
2^C$ such that for each set of configurations $C' \subseteq C$, we have $h(C')=\Interp{ \nu X. \big( (\bigvee_{q \in Q \setminus
    \set{q_F}} q) \wedge (Q_1 \vee \eventually X) \wedge (Q_P \vee
   (Q_1 \wedge \always X)) \vee (q_F \wedge Q_P \wedge \eventually Y)  \vee (q_F \wedge Q_1 \wedge \always Y) \big)}_{\env_0[Y:=C']}$, i.e. it corresponds to the set $\Interp{ \nu X. \big( (\bigvee_{q \in Q \setminus
    \set{q_F}} q) \wedge (Q_1 \vee \eventually X) \wedge (Q_P \vee
   (Q_1 \wedge \always X)) \vee (q_F \wedge Q_P \wedge \eventually Y)  \vee q_F \wedge (Q_1 \wedge \always Y) \big)}$ where $Y$ is interpreted as $C'$. Note that $U$ is then the
least fixpoint of $h$. \\

We first prove that $U \subseteq W^P_S$. Note that $U$ being the least
fixpoint of $h$, by the Knaster-Tarski Theorem we have $U=\bigcap \set{C'
  \subseteq C \mid h(C') \subseteq C'}$. We will hence show that $h(W^P_S)
\subseteq W^P_S$ from which we will get $U \subseteq W^P_S$. For this, 
we consider the function $g_{W^P_S}: 2^C \mapsto
2^C$ such that for each set of configurations $C' \subseteq C$, we have $g_{W^P_S}(C')=\Interp{ (\bigvee_{q \in Q \setminus
    \set{q_F}} q) \wedge (Q_1 \vee \eventually X) \wedge (Q_P \vee
   (Q_1 \wedge \always X)) \vee (q_F \wedge Q_P \wedge \eventually Y)  \vee (q_F \wedge Q_1
   \wedge \always Y) }_{\env_0[Y:=W^P_S,X:=C']}$. The set $V=h(W^P_S)$ is then
 by definition the greatest fixpoint of $g_{W^P_S}$ (so we have as
 well $V=g_{W^P_S}(V)$). Hence we need to show that $V \subseteq
 W^P_S$. For this we will assume that  $c \notin W^P_S$ and show that $c \notin V$.  

Let $c\notin W^P_S$. Hence there exists a strategy
$\sigma \in \Sigma$ such that  $\Outcome{M_S,c,\sigma} \subseteq
\Interp{\always \eventually q_F}$. We reason now
by contradiction assuming that $c \in V$. We will show that either there exists a finite play $c_0 \rightarrow c_1
\rightarrow c_2 \ldots \rightarrow c_k$ in $M_S$ which respects $\sigma$ and with $c_k
\in W^P_S \cap \Interp{q_F}$ or there exists an infinite play $c_0 \rightarrow c_1
\rightarrow c_2 \rightarrow \ldots$ in $\Outcome{M_S,c,\sigma}$ such
that $c_0=c$ and $c_i \notin \Interp{q_F}$ for all $i \in \nat
\setminus \set{0}$. 

We prove in fact by induction  that if $c_0 \rightarrow c_1 \ldots
\rightarrow c_k$ is a finite play in $M_S$ respecting $\sigma$ with
$c_0=c$ and such that $c_i \in V$ and $c_i \notin  W^P_S$  for $i \in [0..k]$, then
there exists $c_{k+1} \in C$ such that either $c_{k+1} \in W^P_S \cap \Interp{q_F}$ or
($c_{k+1} \in V$ and $c_{k+1} \notin \Interp{q_F}$), and such that
$c_0 \rightarrow c_1 \ldots \rightarrow c_{k+1}$ is a play in $M_S$
respecting $\sigma$.  We proceed with the base case. First note
that $c \in V$ and $c \notin
W^P_S$. We recall that $V=g_{W^P_S}(V)$.
\begin{itemize}
\item If $c \in C_1$, then let $c_1=\sigma(c)$. Since $c \in
  g_{W^P_S}(V)$ and $c \in C_1$, we have necessarily that either ($c_1
  \in V$ and $c_1 \notin \Interp{q_F}$) or $c_1\in W^P_S \cap \Interp{q_F}$, and we have
  $c \rightarrow c_1$ is a play in $M_S$ respecting $\sigma$. 
\item If $c \in C_P$, then by definition of $g_{W^P_S}(V)$, there
  exists necessarily $c_1$ such that either ($c_1
  \in V$ and $c_1 \notin \Interp{q_F}$) or $c_1 \in W^P_S \cap \Interp{q_F}$ and $c
  \rightarrow c_1$. Furthermore this allows us to deduce that $c
  \rightarrow c_1$ is a play in $M_S$ respecting $\sigma$. 
\end{itemize}
For the inductive case, the proof works exactly the same way.

From this we deduce that either there exists a finite play $c \rightarrow c_1
\rightarrow c_2 \ldots \rightarrow c_k$ in $M_S$ respecting $\sigma$ with $c_k
\in W^P_S \cap \Interp{q_F}$ or there exists an infinite play $c \rightarrow c_1
\rightarrow c_2 \rightarrow \ldots$ in $\Outcome{M_S,c,\sigma}$ such
that $c_0=c$ and $c_i \notin \Interp{q_F}$ for all $i \in \nat$. We
recall that we have $\Outcome{M_S,c,\sigma} \subseteq
\Interp{\always \eventually q_F}$ and 
proceed by a case analysis to show a contradiction:
\begin{enumerate}
\item If there exists $c_0 \rightarrow c_1
\rightarrow c_2 \ldots \rightarrow c_k$ in $\Outcome{M_S,c,\sigma}$ with $c_k
\in W^P_S$ then it is not possible that all the plays of the form  $c_0 \rightarrow c_1
\rightarrow c_2 \ldots \rightarrow c_k \ldots$ in $\Outcome{M_S,c,\sigma}$
are in $\Interp{\always \eventually q_F}$. In fact, otherwise there would be a
strategy $\sigma'$ from $c_k$ (built from $\sigma$) such that
$\Outcome{M_S,c_k,\sigma'} \subseteq \Interp{\always \eventually q_F}$ which
contradicts the fact that $c_k \in W^P_S$.
\item If there exists an infinite play $c_0 \rightarrow c_1
\rightarrow c_2 \rightarrow \ldots$ in $\Outcome{M_S,c,\sigma}$ such
that $c_0=c$ and $c_i \notin \Interp{q_F}$ for all $i \in \nat$, then we have a
contradiction with the fact that $\Outcome{M_S,c,\sigma} \subseteq
\Interp{\always \eventually q_F}$.
\end{enumerate}
We hence deduce that $c \notin V$, and consequently we have shown that
$U \subseteq W^P_S$.

We now prove that $W^P_S \subseteq U$. To do that we will instead show
that the complement of $U$, denoted by $\overline{U}$, is included in
the complement of $W^P_S$, denoted by $W^1_S$ and which is equal to
the set $\set{c \in C \mid \exists \sigma \in \Sigma  \mbox{ s.t. }
  \Outcome{M_S,c,\sigma} \subseteq \Interp{\always \eventually
    q_F}}$. 

First, note that $\overline{U}$ is equal to $\Interp{\nu Y. \mu X. \big( q_F \vee (Q_P \wedge \always X) \vee (Q_1 \wedge
   (Q_P \vee \eventually X)) \wedge (\neg q_F \vee Q_P \vee \eventually Y)  \wedge (\neg q_F \vee Q_1 \vee \always Y) \big)}$,  which is itself equal
to $\Interp{\nu Y. \mu X. \big[ \big(\neg q_F \wedge \big((Q_P \wedge
  \always X) \vee (Q_1 \wedge \eventually X)\big)\big) \vee \big(q_F \wedge
  \big((Q_P \wedge \always Y) \vee (Q_1  \vee \eventually Y)\big)
  \big)\big]}$. Note that then we have as well that
$\overline{U}= \Interp{\mu X. \big[ \big( \neg q_F \wedge \big((Q_P \wedge
  \always X) \vee (Q_1 \wedge \eventually X)\big)\big) \vee \big(q_F \wedge
  \big((Q_P \wedge \always Y) \vee (Q_1  \vee \eventually Y)\big)
  \big)\big]}_{\env_0[Y:=\overline{U}]}$. We denote by $T$ the set $\Interp{\big(q_F \wedge
  \big((Q_P \wedge \always Y) \vee (Q_1  \vee \eventually Y)\big)
  \big)}_{\env_0[Y:=\overline{U}]}$. Note that we have : $\Interp{\mu X. \big[ \big(\neg q_F \wedge \big((Q_P \wedge
  \always X) \vee (Q_1 \wedge \eventually X)\big) \big)\vee \big(q_F \wedge
  \big((Q_P \wedge \always Y) \vee (Q_1  \vee \eventually Y)\big)
  \big)\big]}_{\env_0[Y:=\overline{U}]} \subseteq \Interp{\mu X. \big[\big( \big((Q_P \wedge
  \always X) \vee (Q_1 \wedge \eventually X)\big)\big) \vee \big(q_F \wedge
  \big((Q_P \wedge \always Y) \vee (Q_1  \vee \eventually
  Y)\big)\big)\big]}_{\env_0[Y:=\overline{U}]}$. Let $U'$ be the set $\Interp{\mu X. \big[\big( \big((Q_P \wedge
  \always X) \vee (Q_1 \wedge \eventually X)\big)\big) \vee \big(q_F \wedge
  \big((Q_P \wedge \always Y) \vee (Q_1  \vee \eventually
  Y)\big)\big)\big]}_{\env_0[Y:=\overline{U}]}$. By adapting the proof
of Lemma \ref{lem:sure-p-pVASS-MDP}, we can deduce that the complement
of this last $U'$ is equal to $\set{c \in C \mid\not \exists \sigma \in \Sigma  \mbox{ s.t. }
  \Outcome{M_S,c,\sigma} \subseteq \Interp{\eventually T}}$ (where
$\Interp{\eventually T}$ denotes the plays that eventually reach the
set $T$) and consequently $U'=\set{c \in C \mid \exists \sigma \in \Sigma  \mbox{ s.t. }
  \Outcome{M_S,c,\sigma} \subseteq \Interp{\eventually T}}$. We have consequently that $\overline{U}
\subseteq \set{c \in C \mid \exists \sigma \in \Sigma  \mbox{ s.t. }
  \Outcome{M_S,c,\sigma} \subseteq \Interp{\eventually T}}$ with $T= \Interp{\big(q_F \wedge
  \big((Q_P \wedge \always Y) \vee (Q_1  \vee \eventually Y)\big)
  \big)}_{\env_0[Y:=\overline{U}]}$. 

We can now prove that
$\overline{U} \subseteq W^1_S$. Let $c \in \overline{U}$. Since   $\overline{U}
\subseteq \set{c \in C \mid \exists \sigma \in \Sigma  \mbox{ s.t. }
  \Outcome{M_S,c,\sigma} \subseteq \Interp{\eventually T}}$, from $c$ Player 1 can surely
reach $T$ and by definition of $T$, when it reaches $T$ first it
is in  $\Interp{q_F}$ and then Player 1 can ensure a successor state
to belong to $\overline{U}$. Hence performing this reasoning
iteratively, we can build a strategy for Player 1 from $c$ to reach
surely infinitely often $\Interp{q_F}$.

\qed
\end{proof}

\section{Proofs of Section \ref{sec:1-VASS-MDP}}

\subsection{Proof of Lemma \ref{lemma-onenpvass-nodeadlock}}

\begin{proof}
  Let $S=\tuple{Q,Q_1,Q_P,T,\tau}$ be a 1-VASS-MDP and $M_S=\tuple{C,C_1,C_P,\trans,p}$ its associated MDP. If a configuration $\tuple{q,\vect{v}} \in C_P$ is a deadlock, then it means that there is no outgoing edge in $S$ from the control state $q$; hence each time a play will reach this configuration, Player 1 will lose, so we can add a self loop without any effect on the counters to this state in order to remove the deadlock. For the states $q \in Q_1$, we add a transition $t$ to the outgoing edge of $q$ which does not modify the counter values and which leads to a new control state with a self-loop, such that if the play reaches a configuration $\tuple{q,\vect{v}}$ which is a deadlock in $S$, in the new game arena the only choice for Player 1 will be to go to this new absorbing state and he will lose as he loses in $S$ because of the deadlock. \qed
\end{proof}

\subsection{Proof of Theorem \ref{thm-sure-1vamdp}}

If we define the two following set of configurations: $V^1_S=\set{c
  \in C \mid \exists \sigma \in \Sigma  \mbox{ such that  }
  \Outcome{M_S,c,\sigma} \subseteq \Interp{\eventually q_F}}$ and
$W^1_S=\set{c \in C \mid \exists \sigma \in \Sigma  \mbox{ such that
  } \linebreak[0] \Outcome{M_S,c,\sigma} \subseteq \Interp{\always \eventually q_F}}$, we have the following result:

\begin{lemma}~~
\label{lem:sure-1-VASS-MDP}
  \begin{itemize}
  \item $V^1_S=\Interp{\mu X. q_F \vee (Q_1 \wedge \eventually X) \vee (Q_P \wedge \always X)}$
  \item  $W^1_S=\Interp{\nu Y.\mu X. \big(q_F \vee (Q_P \wedge \always X) \vee (Q_1 \wedge \eventually X )\big) \wedge \big (\bigvee_{q \in Q \setminus
        \set{q_F}} q \vee Q_1 \vee (Q_P \wedge \always Y)\big) \wedge \big( \bigvee_{q \in Q \setminus
        \set{q_F}} q  \vee Q_P \vee \eventually Y \big)}$
  \end{itemize}
\end{lemma}

\begin{proof}
  Since we are looking at deadlock free VASS-MDP, we can reuse the formula given by Lemma \ref{lem:sure-p-pVASS-MDP} and \ref{lem:almostsure-p-pVASS-MDP}. In fact, note that we have $V^1_S=\overline{V^P_S}$ and $W^1_S=\overline{W^P_S}$. Hence by taking the complement formulae of $\guardmucalcul$, we obtain the desired result. For $V^1_S$, from the formula describing $V^P_S$, we obtain the formula $\mu X. q_F \vee (Q_1 \wedge \eventually X) \vee \big(Q_1 \wedge (Q_P \vee \always X)\big)$ which is equivalent to $\mu X. q_F \vee (Q_1 \wedge \eventually X) \vee (Q_P \wedge \always X)$. For $W^1_S$, from the formula describing $W^P_S$, we obtain the formula $\nu Y.\mu X. \big(q_F \vee (Q_P \wedge \always X) \vee (Q_1 \wedge (Q_P \vee \eventually X) )\big) \wedge \big (\bigvee_{q \in Q \setminus
        \set{q_F}} q \vee Q_1 \vee \always Y \big) \wedge \big( \bigvee_{q \in Q \setminus
        \set{q_F}} q  \vee Q_P \vee \eventually Y \big)$ which is equivalent to $\nu Y.\mu X. \big(q_F \vee (Q_P \wedge \always X) \vee (Q_1 \wedge \eventually X )\big) \wedge \big (\bigvee_{q \in Q \setminus
        \set{q_F}} q \vee Q_1 \vee (Q_P \wedge \always Y)\big) \wedge \big( \bigvee_{q \in Q \setminus
        \set{q_F}} q  \vee Q_P \vee \eventually Y \big)$. \qed
\end{proof}

\subsection{Closed $\guardmucalcul$-formulae manipulate upward closed-sets}

We will say that an environment $\env : \Var \rightarrow 2^C$ is upward-closed if for each variable $X \in \Var$, $\env(X)$ is upward-closed (we take as order $\leq$ for the configurations, the classical one such that $\tuple{q,\vect{v}} \leq \tuple{q',\vect{v}'}$ iff $q=q'$ and $\vect{v} \leq \vect{v}'$) . Whereas it is not true that for any VASS, any formula $\phi \in \guardmucalcul$ and any upward closed environment $\env : \Var \rightarrow 2^C$ the set $\Interp{\phi}_{\env}$ is upward closed, we now prove that on single-sided VASS this property holds.

\begin{lemma}\label{lem-phi-upcl}
For any formula $\phi \in \guardmucalcul$ and any upward closed environment $\env$, $\Interp{\phi}_{\env}$ evaluated over the configurations of the $(Q_1,Q_P)$-single-sided VASS $\tuple{Q,T}$ is an upward closed set.
\end{lemma}

\begin{proof}
The proof is by induction on the length of the formula $\phi$. For formulae of the form $q$, the result is due to the fact that the set of considered regions is  upward closed. For formulae of the form $X$, the result comes from the assumption on the considered environment. For formulae of the form $\phi \wedge \psi$ and $\mu X.\phi$ the result can be obtained using the induction hypothesis and the fact that the intersection of upward closed sets is an upward-closed set. For formulae of the form $\phi \vee \psi$ and $\nu X.\phi$ the result can be obtained using the induction hypothesis and the fact that the union of upward closed sets is an upward closed set. 

Now we consider formulae of the form $\Diamond \phi$ assuming that for any upward closed environment $\env$, $\Interp{\phi}_{\env}$ is an upward closed set. Let $c_1 \in \Interp{\Diamond \phi}_{\env}$. Then there exists $c'_1 \in \Interp{\phi}_{\env}$ such that $c_1 \trans c'_1$. Let $c_2 \in C$ such that $c_1 \leq c_2$. Since we are considering VASS, we now that there exists $c'_2 \in C$ such that $c_2 \trans c'_2$ and $c'_1 \leq c'_2$. Since $\Interp{\phi}_{\env}$ is upward closed, we have $c'_2 \in \Interp{\phi}_{\env}$, hence $c_2$ belongs to $\Interp{\Diamond \phi}_{\env}$. This proves that $\Interp{\Diamond \phi}_{\env}$ is upward closed.

Now we consider formulae of the form $Q_P \wedge \Box \phi$ assuming that for
any upward closed environment $\env$, $\Interp{\phi}_{\env}$ is an upward
closed set. Let $c_1 \in \Interp{Q_P \wedge \Box \phi}_{\env}$. Then for all
$c'_1 \in C$ such that $c_1 \trans c'_1$, we have $c'_1 \in
\Interp{\phi}_{\env}$. Note that $c_1=(q_1,\vect{v}_1)$ with $q_1 \in
Q_P$. Let $c_2 \in C$ such that $c_1 \leq c_2$. By definition of the order
$\leq$ on the set of configurations, we have $c_2=(q_1,\vect{v}_2)$ with
$\vect{v}_1 \leq \vect{v}_2$. Let $c'_2$ such that $c_2 \trans c'_2$. Since
$q_1 \in Q_P$, by the definition of single-sided VASS, we know that the transition which leads from $c_2$ to $c'_2$ can also be taken from $c_1$ (this is because the outgoing transitions from control states in $Q_P$ do not modify the counter values), hence there exists $c'_1 \in C$ such that $c_1 \trans c'_1$ and since $c_1 \leq c_2$, we have $c'_1 \leq c'_2$. Furthermore, we have $c'_1 \in \Interp{\phi}_{\env}$ and since, by induction, this last set is upward closed, we deduce $c'_2 \in \Interp{\phi}_{\env}$. This allows us to conclude that $c_2$ belongs to $\Interp{Q_1 \wedge \Box \phi}_{\env}$ which is hence an upward closed set. \qed
\end{proof}

\subsection{Proof of Lemma \ref{lemma-N-strat}}

\begin{proof}
We consider the function $h : 2^C \rightarrow 2^C$ which associates to each set of configurations $C' \subseteq C$ the set $h(C')=\Interp{q_F \vee \InvPre{X,Y}\big}_{\env_0[X:=E,Y:=C']}$. We define a sequence of sets $(F_i)_{i \in \nat}$ included in $C$ as follows:
\begin{itemize}
\item $F_0=\emptyset$
\item for all $i \in \nat$, $F_{i+1}=F_i \cup h(F_i)$
\end{itemize}
Using Lemma \ref{lem-phi-upcl} and the fact that the union of upward closed set is an upward closed set, we can prove that $F_i$ is upward closed for all $i \in \nat$. Furthermore, we have that $F_i \subseteq F_{i+1}$ for all $i \in \nat$. Since $(F_i)_{i \in \nat}$ is an increasing sequence of upward closed sets included in $C$ and since $(C,\leq)$ is a wqo, from the theory of wqo, we know that there exists $N \in \nat$ such that for all $i \geq N$, $F_i=F_{i+1}$. 

We consider also the function $g : 2^C \rightarrow 2^C$, which associates to each set of configurations $C' \subseteq C$ the set $g(C')=\Interp{\mu Y. \big ( q_F \vee \InvPre{X,Y}\big )}_{\env_0[X:=C']}$. By definition, $E$ is the greatest fixpoint of the function $g$, hence $E=\Interp{\mu Y. \big ( q_F \vee \InvPre{X,Y}\big )}_{\env_0[X:=E]}$. $E$ is then also the least fixpoint of the function $h$ and consequently, by definition of the sequence $(F_i)_{i \in \nat}$, we know that $E=\bigcup_{i \in \nat} F_i$. This allows us to deduce that $E=F_N$. We point out that $F_1=\Interp{q_F}$.

 We now define a strategy $\sigma$ for Player 1 which will be memoryless on $E$ (i.e. the strategy will only depend on the current configuration) and is a function $\sigma: E \cap C_1 \mapsto C$ (note that since with this strategy all the plays starting from $E$ will stay in $E$, we do not need to define it precisely on the entire set $C_1$ and we assume that, on the set $C_1 \setminus E$, the strategy can choose any one of the possible successor configurations). Let $c \in E$. Then we denote by $j \in [2..N]$ the smallest index such that $c \in F_{j}$ and $c \notin F_{j-1}$, we then define $\sigma(c)$ as being the configuration $c'$ such that $c \trans c'$ and $c' \in F_{j-1}$ (by definition of the sequence ($(F_i)_{i \in \nat}$, such a $c'$ necessarily exists). If $c$ belongs to $F_1$, then the strategy chooses any one of the possible successors. 
 
Let $c_0 \in E$. We  show  that there exists a play $c_0 \cdot c_1 \cdot c_2 \cdot \ldots$ in $\Outcome{M_S,c_0,\sigma}$ that satisfies the three properties of the lemma. In fact, we consider the play such that in all configurations $c \in E \cap C_P$ and if $j \in [2..N]$ is the smallest index such that $c \in F_{j}$ and $c \notin F_{j-1}$, Player P chooses $c'$ such that $c \trans c'$ and $c' \in F_{j-1}$. Hence in this play it is obvious that in less than $N$ steps, the play will reach a configuration in $F_1=\Interp{q_F}$ and the points 2. and 3. also hold for this play by definition of the sets $(F_i)_{i \in \nat}$.\qed
\end{proof}

\subsection{Proof of Lemma \ref{lem-probab-formula}}

We now prove that the set $E$ is included in $V^1_{AS}$. Techniques we used here are quite similar than the one presented in \cite{abdulla-decisive-07} to prove decidability of the sure reachability problems in probabilistic VASS (without nondeterminism).

\begin{lemma}\label{lem-reach-Ew0}
$E \subseteq V^1_{AS}$
\end{lemma}

\begin{proof}
We consider the integer $N \in \nat$ and the strategy $\sigma$ of Player 1 given by Lemma \ref{lemma-N-strat}. For each $q \in Q_P$, let $\mathtt{Out}(q)=\set{(q,\vect{z},q') \in T \mbox{ for some } q' \in Q \mbox{ and } \vect{z} \in \rel^n}$ be the set of transitions going out of $q$. We also denote by $L_q$ the cardinality of $\mathtt{Out}(q)$, by $W_q$ the sums $\Sigma_{t \in  \mathtt{Out}(q)} \tau(t)$ and finally $\mathit{Min}_q$ is the minimal element of $\set{\tau(t) \mid t \in \mathtt{Out}(q)}$. By definition of VASS-MDP, we know that for a configuration $\tuple{q,\vect{v}} \in C_P$, for any configuration $c' \in C$ such that $\tuple{q,\vect{v}} \trans c'$, we have $p(q,\vect{v})(c') \geq \frac{\mathit{Min}_q}{L_q \cdot W_q}$. We denote by $\beta$ the minimal element of the set $\set{\frac{\mathit{Min}_q}{L_q \cdot W_q} \mid q \in Q_P}$. Then for any configuration $c \in C_P$ and $c' \in C$ such that $c \trans c'$, we have $p(c)(c') \geq \beta$. Note that necessarily $\beta > 0$. Let $c_0 \in E \setminus \Interp{q_F}$ and let $c_0 \cdot c_1 \cdot c_2 \cdots $ be a play in $\Outcome{M_S,c_0,\sigma}$ such that for all $i \in \nat$, $c_i \notin \Interp{q_F}$. From Lemma \ref{lemma-N-strat}, we know that there exists $N \in \nat$ such that, for all $i \in \nat$, $c_i \in E$, we have that $\Prob(M_S,c_i,\sigma,\Interp{\eventually q_F}) \geq \beta^N$. This allows us to deduce that the probability of never visiting $q_F$ from $c_0$ following $\sigma$ is smaller than $(1-\beta^N)^\infty$ and since $\beta>0$, we deduce that this probability is equal to $0$. Consequently $\Prob(M_S,c_0,\sigma,\Interp{\eventually q_F}) =1$ and $c_0 \in V^1_{AS}$. \hfill $\Box$ \end{proof}

We now prove the opposite direction. For this we use a technique similar to the ones presented in the proof of Lemma 5.29 in \cite{bertrand-phd-06}.
\begin{lemma}\label{lem-reach-w0E}
$V^1_{AS} \subseteq E$
\end{lemma}

\begin{proof}
Let $c_0 \in V^1_{AS}$. So there exists a strategy $\sigma$ for Player 1 such that $\Prob(M_S,\linebreak[0]c_0,\sigma,\Interp{\eventually q_F})=1$. Let $D$ be the following set of configurations: $\set{c \in C \mid \exists c_0 \cdot c_1 \cdots \in \Outcome{M_S,c_0,\sigma} \mbox{ s.t. } \exists i \in \nat \mbox{ for which } c_i = c \mbox{ and } \forall 0 \leq j < i. c_j \notin \Interp{q_F}}$. Clearly $c_0$ belongs to $D$. We will show that $D \subseteq E$. 

We consider the two functions $g,h : 2^C \rightarrow 2^C$ such that for each set of configurations $C' \subseteq C$, we have $g(C')=\Interp{\mu Y. \big ( q_F \vee \InvPre{X,Y}\big )}_{\env_0[X:=C']}$ and $h(C')=\Interp{q_F \vee \InvPre{X,Y}\big}_{\env_0[X:=D,Y:=C']}$. We define the following sequence of configurations $(F^D_i)_{i \in \nat}$ such that:
\begin{itemize}
\item $F^D_0=\emptyset$,
\item for all $i \in \nat$, $F^D_{i+1}=F^D_i \cup h(F^D_i)$
\end{itemize}

Let $c \in D$, since $ \Prob(M_S,c_0,\sigma,\Interp{\eventually q_F})=1$, we know that there exists a strategy $\sigma'$ such that $\Prob(M_S,c,\sigma',\Interp{\eventually q_F})=1$ (otherwise we would have $ \Prob(M_S,c_0,\sigma,\Interp{\eventually q_F})\linebreak[0]<1$). Hence there is a play in $c \cdot c'_1 \cdot c'_2 \cdots \in \Outcome{M_S,c,\sigma'}$ for which there exists $i \in \nat$ satisfying $c'_i \in  \Interp{q_F}$ and  $c'_j \notin \Interp{q_F}$ for all $0 \leq j <i$. Note also that by definition of $D$, for all $0 \leq j \leq i$, we have that $c'_j$ belongs to $D$ and if $c'_j \in C_P$ then for all $c'' \in C$ such that $c'_j \trans c''$, we have $c'' \in D$. Hence there exists $0 \leq k \leq i+1$, such that $c \in F^D_k$. Since $g(D)=\Interp{\mu Y. \big ( q_F \vee \InvPre{X,Y}\big )}_{\env_0[X:=D]}$, we have $g(D)=\bigcup_{i \in \nat} F^D_i$. Consequently, $c \in g(D)$ and we have $D \subseteq g(D)$.

We know that $E$ is the greatest fixpoint of $g$, then by Knaster-Tarski Theorem we know that $E=\bigcup \set{C' \in C \mid C' \subseteq g(C')}$, and hence $D \subseteq E$. Since $c_0 \in D$, we deduce that $c_0 \in E$. \hfill $\Box$
\end{proof}

\subsection{Proof of Lemma \ref{lem-probab-formula-repeat}}
We denote by $F$ the set $\Interp{\nu X. \InvPre{X,\mu Y.( q_F \vee \InvPre{X,Y}) \big}}_{\env_0}$. Since $\nu X. \InvPre{X,\mu Y.( q_F \vee \InvPre{X,Y})}$ is a formula of $L^{sv}_\mu$ interpreted over the $(Q_1,Q_P)$-single-sided VASS $\tuple{Q,T}$, from Lemma \ref{lem-phi-upcl}, we know that $F$ is an upward-closed set. 

We first prove that $W^1_{AS}$ is included in $F$.

\begin{lemma}
$W^1_{AS} \subseteq F$
\end{lemma}

\begin{proof}
Let $c_0 \in W^1_{AS}$. So there exists a strategy $\sigma$ for Player 1 such that $\Prob(M_S,\linebreak[0]c_0,\sigma,\Interp{\always \eventually q_F})=1$. Let $T$ be the following set of configurations $\set{c \in C \mid \exists c_0 \cdot c_1 \cdots \in \Outcome{M_S,c_0,\sigma} \mbox{ s.t. } \exists i \in \nat \mbox{ for which } c_i = c}$. Necessarily, we have $T \cap \Interp{q_F} \neq \emptyset$ and $c_0 \in T$. We consider the function $g : 2^C \rightarrow 2^C$ such that for each set of configurations $C' \subseteq C$, we have $g(C')=\Interp{\InvPre{X,\mu Y.( q_F \vee \InvPre{X,Y})}}_{\env_0[X:=C']}$. We will prove that $T \subseteq g(T)$.

Let $c \in T$. We have necessarily that there exists a strategy $\sigma'$ for Player 1 such that   $\Prob(M_S,c,\sigma',\Interp{\always \eventually q_F}) > 0$, otherwise we would have that $\Prob(M_S,c_0,\sigma,\Interp{\always \eventually q_F})< 1$. Hence there exists a play in $M_S$ respecting $\sigma$ of the form:
$$
c_0\cdot c_1 \cdots c_k \cdots c \cdot c'_1 \cdots c'_m
$$
with $m \geq 1$ (that is $c \neq c'_m$) and $c'_m \in \Interp{q_F}$. By definition of $T$, for each $1 \leq i \leq m$, if $c'_i \in C_P$, for all configurations $c'' \in C$ such that $c'_i \trans c''$ we have,  $c'' \in T$. Using a similar reasoning to that done in Lemma \ref{lem-reach-w0E}, we deduce that for all $1 \leq i \leq m$, we have $c'_i \in \Interp{\mu Y.( q_F \vee \InvPre{X,Y})}_{\env_0[X:=T]}$. Furthermore if $c \in C_P$, for all configurations $c'' \in C$ such that $c \trans c''$ we have also $c'' \in T$, hence we have that $c \in g(T)$ (using the definition of $\InvPre{X,Y}$). This implies $T \subseteq g(T)$.

Since $F$ is the greatest fixpoint of $g$, then by Knaster-Tarski Theorem we know that $F=\bigcup \set{C' \in C \mid C' \subseteq g(C')}$, so $T \subseteq F$. Since $c_0 \in T$, we deduce that $c_0 \in F$. \hfill $\Box$
\end{proof}

We now prove the left to right inclusion.

\begin{lemma}
$F \subseteq W^1_{AS}$
\end{lemma}

\begin{proof}
We consider the function $h : 2^C \rightarrow 2^C$ which associates to each set of configurations $C' \subseteq C$ the set $h(C')=\Interp{q_F \vee \InvPre{X,Y}\big}_{\env_0[X:=F,Y:=C']}$. We define a sequence of sets $(F_i)_{i \in \nat}$ included in $C$ as follows:
\begin{itemize}
\item $F_0=\emptyset$
\item for all $i \in \nat$, $F_{i+1}=F_i \cup h(F_i)$
\end{itemize}
As for the proof of Lemma \ref{lemma-N-strat}, we know that there exists $N \in \nat$ such that for all $i \geq N$, $F_i=F_{i+1}$ and that $\Interp{\mu Y. \big ( q_F \vee \InvPre{X,Y}\big )}_{\env_0[X:=F]}=F_N$. Since $F=\Interp{\nu X. \InvPre{X,\mu Y.( q_F \vee \InvPre{X,Y})}}_{\env_0}$, using again fixpoint theory, we know that $F=\Interp{\InvPre{X,\mu Y.( q_F \vee \InvPre{X,Y})}}_{\env_0[X:=F]}$ and consequently $F=\Interp{\InvPre{X,Y}}_{\env_0[X:=F,Y:=F_N]}$.

We now define a strategy $\sigma$ for player $0$ which will be memoryless on $F$ (i.e. the strategy will only depend of the current configuration). The strategy $\sigma$ will be described as a function $\sigma: F \cap C_1 \mapsto C$ (note that since with this strategy all the plays starting from $F$ will stay in $F$, we do not need to define it precisely on the entire set $C_1$; we assume that on the set $C_1 \setminus F$, the strategy can choose any one of the possible successor states). Let $c \in C_1 \setminus F$ and consider the following two cases. If $c \in F \setminus F_N$ or $c \in F_1$, we define $\sigma(c)$ as being a configuration $c' \in F_N \cap F$ such that $c \trans c'$. By definition of $F$ such a configuration necessarily exists. If $c \in F_N \setminus F_1$ we denote by $j \in [2..N]$ the smallest index such that $c \in F_{j}$ and $c \notin F_{j-1}$. We then define $\sigma(c)$ as being the configuration $c'$ such that $c \trans c'$ and $c' \in F_{j-1} \cap F$ (by definition of $F$ and of the sequence $(F_i)_{i \in \nat}$, such a $c'$ necessarily exists).

By construction of the strategy $\sigma$ and using a similar reasoning to the one performed in the proof of Lemma \ref{lem-reach-Ew0}, we can prove that for all $c \in F$, we have $\Prob(M_S,c,\sigma,\Interp{\eventually q_F}) =1$. We would like now to prove that for each $c \in F$, $\Prob(M_S,c,\sigma,\Interp{\always \eventually q_F}) =1$.

We denote by $\Interp{\always \neg q_F}$ the set of infinite plays $c_0 \cdot c_1 \cdots$ of $M_S$ such that for all $i \in \nat$, $c_i \notin \Interp{q_F}$. Then for each $c \in F$, we have $\Prob(M_S,c,\sigma,\Interp{\always \neg q_F}) =0$. We also use the notation $\Interp{\eventually \always \neg q_F}$ the set of infinite play $c_0 \cdot c_1 \cdots$ of $M_S$ for which there exists $i \in \nat$, such that $c_j \notin \Interp{q_F}$ for all $j \geq i$. Then for each $c \in F$, we have $\Prob(M_S,c,\sigma,\Interp{\always \eventually q_F}) =1- \Prob(M_S,c,\sigma,\Interp{\eventually \always \neg q_F})$. We will now prove that for each $c \in F$, $\Prob(M_S,c,\sigma,\Interp{\eventually \always \neg q_F})=0$.

For $c \in C$, $i \in \nat$ and $d \in C \setminus \Interp{q_F}$, let $\Pi_{c-i-d}$ be the set of finite plays of $M_S$ of the form $c_0 \cdot c_1 \cdots c_k$ such that:
\begin{itemize}
\item $k>i$;
\item $c_0=c$, $c_k=d$ and $c_{k-1} \in \Interp{q_F}$;
\item the set $G=\set{j \in \set{0,\ldots,k} \mid c_j \in \Interp{q_F}}$ has $i$ elements.
\end{itemize}
This represents the set of finite plays starting at $c$ ending at configuration $d$ and that passes exactly through $i$ configurations in $\Interp{q_F}$.


For $c \in F$ and $i \in \nat$ and $d \in F \setminus \Interp{q_F}$, we define $\Delta_{c-i-d}$ the set of infinite plays of the form $\rho \cdot \rho'$ where $\rho$ is a finite play in $\Pi_{c-i-d}$ and $\rho$ is an infinite play in $\Interp{\always \neg q_F}$. Let $\Delta_{c-i}=\bigcup_{d \notin \Interp{q_F}} \Delta_{c-i-d}$. Intuitively $\Delta_{c-i}$ is the set of infinite plays starting from $c$ which revisits $\Interp{q_F}$ exactly $i$ times. For $c \in F$, it is straightforward to check that:
\begin{description}
\item[(1)] $\forall i \in \nat$, $\forall d_1,d_2 \in F \setminus \Interp{q_F}$ such that $d_1 \neq d_2$, $\Delta_{c-i-d_1} \cap \Delta_{c-i-d_2}=\emptyset$
\item[(2)] $\forall i,j \in \nat$ such that $i \neq j$, $\Delta_{c-i} \cap \Delta_{c-j}=\emptyset$
\item[(3)] $\forall i \in \nat$, $\forall d \in F \setminus \Interp{q_F}$, $\Prob(M_S,c,\sigma,\Delta_{c-i-d})=P(\Pi_{c-i-d}) \Prob(M_S,d,\sigma,\Interp{\always \neg q_F})$
\end{description}

Hence, for all $c \in F$ and $i \in \nat$, we have:
$$
\begin{array}{rcl}
\Prob(M_S,c,\sigma,\Delta_{c-i})& = &\sum_{d \in F \setminus \Interp{q_F}} 
\Prob(M_S,c,\sigma,\Delta_{c-i-d}) \\
& = & \Sigma_{d \in F \setminus \Interp{q_F}} P(\Pi_{c-i-d}) \Prob(M_S,d,\sigma,\Interp{\always \neg q_F}) \\
& = & 0
\end{array}
$$
where the first equality holds by (1) and by the fact that all the configurations reached from $c$ following $\sigma$ belongs to $F$ (by definition of $\sigma$ and $F$); the second equality follows from (3) and the last equality from the fact that for all $d \in F$, $\Prob(M_S,d,\sigma,\Interp{\always \neg q_F}) =0$.

Finally, we have for all $c \in F$:
$$
(\Outcome{M_S,c,\sigma} \cap \Interp{\eventually \always \neg q_F}) \subseteq (\Outcome{M_S,c,\sigma} \cap \bigcup_{i \in \nat} \Delta_{c-i})
$$
 This is due to the fact that if an infinite play belongs to $\Outcome{M_S,c,\sigma} \cap \Interp{\eventually \always \neg q_F}$, then it will pass only a finite number of times through $\Interp{q_F}$. From this inclusion, the previous equality and using (2), we deduce that $\Prob(M_S,c,\sigma,\Interp{\eventually \always \neg q_F}) \linebreak[0]\leq \sum_{i \in \nat} \Prob(M_S,c,\sigma,\Delta_{c-i})=0$. Hence, for all $c \in F$, we have $\Prob(M_S,c,\sigma,\Interp{\always \eventually q_F}) =1$, which allows us to conclude that $F \subseteq W^1_{AS}$. \qed
\end{proof}

\subsection{Proof of Lemma \ref{lem:alg-limitsure-termination}}

\begin{proof}
Algorithm~\ref{alg:limitsure} explores an unfolding of the computation tree of
$S$, which is finitely branching since $|T|$ is finite.
The number of counters is fixed, and therefore, by Dickson's Lemma,
$(\mathbb{N}^d,\preceq)$ is a well quasi ordering.
Therefore, on every branch we eventually satisfy either the condition of line 19 or
of line 7.
In the former case, a loop in the derived system $S'$ is created, and the
exploration of the current branch stops.
In the latter case, a finitary description of a new colored (possibly
infinite-state) subsystem is added to $S'$
by adding finitely many states, transitions and configurations to $Q'$, $T'$
and $X'$, respectively. Also in this case, the exploration of the current branch stops.
Since the exploration is finitely branching, and every branch eventually
stops, the algorithm terminates. \qed
\end{proof}

\subsection{Proof of Lemma \ref{lem:limitsure-correctness-algo}}

\begin{lemma}\label{lem:limitsure_oldtonew}
$\mathbb{P}^+(M_S, c_0, \Interp{\eventually X}) = 1 \implies 
\mathbb{P}^+(M_{S'}, c_0', \Interp{\eventually X'}) = 1$.
\end{lemma}
\begin{proof}
Let us assume that $\mathbb{P}^+(M_S, c_0, \Interp{\eventually X}) = 1$. 
Therefore, there exists a family of strategies that make the probability of
reaching $X$ arbitrarily close to $1$. 
In other words, $\forall \epsilon,\exists \sigma_{\epsilon}, \mathbb{P}(M_S, c_0, \sigma_{\epsilon},  \Interp{\eventually X}) \geq 1- \epsilon$. 
For every $\epsilon >0$ we use the strategy $\sigma_\epsilon$ of player 1 on
$M_S$ to construct a copycat strategy $\sigma'_\epsilon$ for the game on $M_{S'}$ 
that starts in $c'_0=(q_0,\textbf{0} )$, such that it achieves
$\Interp{\eventually X'}$ 
with probability $\geq 1- \epsilon$. 

The strategy $\sigma'_\epsilon$ will use the same moves on $M_{S'}$ as
$\sigma_\epsilon$ on $M_S$, which is possible due to the way how $M_{S'}$ 
is constructed from $M_S$ by Algorithm~\ref{alg:limitsure}.
By construction, for every reachable configuration in $M_S$ there is a
corresponding configuration in $M_{S'}$, and this correspondence can be maintained
stepwise in the moves of the game.

For the initial uncolored part of $M_{S'}$, this is immediate, since $S'$ is
derived from the unfolding of the game tree of $S$.
The correspondence is expressed by the function $\lambda$. Each current state of $M_{S'}$ is labeled
by the corresponding current configuration of $M_S$.

In the colored subsystems, the corresponding configuration in system $M_{S'}$ 
is a projection of a configuration in $M_S$.
For any transition $t \in T$ that is controlled by player 1 from a
configuration in $M_S$, there exists a transition $t' \in T'$ that belongs to
player 1 in the corresponding configuration in $M_{S'}$, such that this transition
leads to the corresponding state.
This is achieved by the projection and the fact that the 1-VASS-MDP game is monotone
w.r.t.\ player 1, i.e., larger configurations always benefit the player (by
allowing the same moves or even additional moves).

We now show a property on how probabilistic transitions in $M_S$ and $M_{S'}$
correspond to each other:
For every probabilistic transition $t \in T$ from a configuration in $M_S$, 
there exists a probabilistic transition $t' \in T'$ in the corresponding
configuration in $M_{S'}$, and vice-versa, such that these transitions have the same
probability.
In particular, a configuration in $M_{S'}$ does not allow any additional
probabilistic transitions compared to its corresponding configuration in $M_S$
(though it may allow additional transitions controlled by player 1).

The first part of this statement follows from the monotonicity of the
projection function and the monotonicity of the transitions 
w.r.t.\ the size of the configurations.
For the second part we need to show that for every probabilistic transition  
$t'=\tuple{col_i(x), proj_k(op), col_i(y)} \in T'$ from a configuration in $M_{S'}$, 
there exists a probabilistic transition $t=\tuple{x,\vect{0},y} \in T_P$ in the
corresponding configuration in $S$, such that the probabilities of these
transitions are equal. 
This latter fact holds only because we are considering 1-VASS-MDP, where only
the player can change the counters, 
whereas the probabilistic transitions can only change the control-states.
I.e., the `larger' projected configurations in $M_{S'}$ do not enable additional
probabilistic transitions, since in 1-VASS-MDP these only depend on the
control-state.

Therefore, by playing in $M_{S'}$ using strategy $\sigma'_\epsilon$ with the same
moves as $\sigma_\epsilon$ plays in $M_S$, we reach the same corresponding
configurations in $M_{S'}$ with the same probability values as in $M_S$.
Since the definition of the target set $X'$ in $S'$ includes all
configurations corresponding to configurations in $X$ on $S$,
it follows from 
$\mathbb{P}(M_S, c_0, \sigma_{\epsilon},  \Interp{\eventually X}) ) \geq 1-\epsilon$
that 
$\mathbb{P}(M_{S'}, c_0', \sigma'_{\epsilon},  \Interp{\eventually X'}) \geq
1- \epsilon$.
Since, by assumption above, this holds for every $\epsilon >0$,
we obtain $\mathbb{P}^+(M_{S'}, c'_0,\Interp{\eventually X'}) ) = 1$.
 \qed 
\end{proof}

\begin{lemma}\label{lem:limitsure_newtoold}
$\Prob^+(M_{S'}, c_0', \Interp{\eventually X'}) = 1 \implies \Prob^+(M_S, c_0, \Interp{\eventually X}) = 1$.
\end{lemma}
\begin{proof}
We use the assumed family of strategies on $M_{S'}$ that witnesses the property
$\Prob^+(M_{S'}, c_0', \Interp{\eventually X'}) = 1$ to 
synthesize a family of strategies on $M_S$ that witnesses
$\Prob^+(M_S, c_0, \Interp{\eventually X}) = 1$.

First we establish some basic properties of the system $S'$. It is a 1-VASS-MDP
of dimension $d-1$ with initial configuration $c_0'$, and consists of several
parts.
The initial uncolored part induces a finite-state MDP. Moreover, $S'$ contains finitely many
subsystems of distinct colors, where each subsystem is a 1-VASS-MDP
of dimension $d-1$ obtained from $S$ by projecting out one component of the
integer vector. For color $i$, let $k(i)$ be the projected component of the
vector (see line 9 of the algorithm).
Each colored subsystem of dimension $d-1$ induces an MDP that may be
infinite-state (unless $d=1$, in which case it is finite-state). 

Note that colored subsystems are not reachable from each
other, i.e., a color, once reached, is preserved.
Each colored subsystem has its own initial configuration (created in lines 11-12 of
Alg.~\ref{alg:limitsure}). Let $m$ be the number of colors in $S'$ and
$r_i$ the initial configuration of the subsystem of color $i$ 
(where $0 \le i \le m-1$).

Let's now consider only those colored subsystems in which the target set $X'$
can be reached limit-surely, i.e., let 
$J = \{i:0\le i \le m-1\ |\ \Prob^+(M_{S'}, r_i, \Interp{\eventually X'}) = 1\}$
be the set of good colors and let
$R = \{r_j\ |\ j \in J\}$, and $\bar{R} = \{r_j\ |\ j \notin J\}$.

Further, let $X_f'$ be the restriction of $X'$ to the finite uncolored part of
$S'$ (i.e., only those parts added in line 25 of Alg.~\ref{alg:limitsure}). 

We now establish the existence of certain strategies in subsystems of $S'$.
These will later serve as building blocks for our strategies on $M_S$.

Since we assumed that $\Prob^+(M_{S'}, c_0', \Interp{\eventually X'}) = 1$, 
there exists a family of strategies that makes the probability of reaching
$X'$ arbitrarily close to one. In particular, they must also make the 
probability of reaching configurations in $\bar{R}$ arbitrarily close to zero.
Thus we obtain $\Prob^+(M_{S'}, c_0', \Interp{\eventually X_f' \cup R}) = 1$,
i.e., we can limit-surely reach $X_f' \cup R$. Since, for this objective, 
only the finite uncolored part of $M_{S'}$ is relevant, this is a problem for a
finite-state MDP and limit-surely and almost-surely coincide.
So there exists a partial strategy $\astrat$, for the uncolored part of $M_{S'}$,
such that, starting in $c_0'$, we almost-surely reach $X_f' \cup R$, i.e.,
$\Prob(M_{S'}, c_0', \astrat, \Interp{\eventually X_f' \cup R}) = 1$.

In each of the good colored subsystems we can limit-surely reach $X'$, i.e., 
for every $r_j \in R$ we have 
$\Prob^+(M_{S'}, r_i, \Interp{\eventually X'}) = 1$. 
So for every $\epsilon >0$ there exists a strategy $\astrat_i^\epsilon$ such that
$\Prob(M_{S'}, r_i, \astrat_i^\epsilon, \Interp{\eventually X'}) \ge 1-\epsilon$.
Consider the computation tree of the game on $M_{S'}$ from $r_i$ when playing according to
$\astrat_i^\epsilon$ and its restriction to some finite depth $d$.
Let $\Prob_d(M_{S'}, r_i, \astrat_i^\epsilon, \Interp{\eventually X'})$ be the
probability that the objective is reached already during the first $d$ steps
of the game. We have
$\Prob_d(M_{S'}, r_i, \astrat_i^\epsilon, \Interp{\eventually X'}) \le \Prob(M_{S'}, r_i, \astrat_i^\epsilon, \Interp{\eventually X'})$, 
but
$\lim_{d \rightarrow \infty} \Prob_d(M_{S'},
r_i, \astrat_i^\epsilon, \Interp{\eventually X'}) = 
\Prob(M_{S'}, r_i, \astrat_i^\epsilon, \Interp{\eventually X'}) \ge
1-\epsilon$.
Thus for every color $i \in J$ and every $\epsilon >0$ there exists a
number $d(i,\epsilon)$ s.t. 
$\Prob_{d(i,\epsilon)}(M_{S'}, r_i, \astrat_i^\epsilon, \Interp{\eventually
X'}) \ge 1-2\epsilon$.

Since configurations of $M_{S'}$ are obtained by projecting configurations
of $M_S$, we can go the reverse direction by replacing the missing component in
an $M_{S'}$ configuration by a given number. Given an $M_{S'}$-configuration $r_i$ and
a number $d(i,\epsilon)$ we obtain an $M_S$-configuration $s_i(d(i,\epsilon))$ by replacing the
missing $k(i)$-th component of $r_i$ by $d(i,\epsilon)$.
Let $\alpha \in \nat$ be the maximal constant appearing in any transition in
$S$, i.e., the maximal possible change in any counter in a single step.
Since a single step in $M_S$ can only change a counter by $\le \alpha$, the $k(i)$-th
component of $s_i(\alpha * d(i,\epsilon))$ cannot be exhausted during the first $d(i,\epsilon)$ steps
of the game on $M_S$ starting at $s_i(\alpha * d(i,\epsilon))$.
Thus we can use the same strategy $\astrat_i^\epsilon$ in the game from
$s_i(\alpha * d(i,\epsilon))$
on $M_S$ and obtain 
$\Prob(M_S, s_i(\alpha * d(i,\epsilon)), \astrat_i^\epsilon, \Interp{\eventually X}) \ge 1-2\epsilon$.
Intuitively, the number $\alpha * d(i,\epsilon)$ is big enough to allow playing the
game for sufficiently many steps to make the probability of success close to $1$. 

Using the strategy $\astrat$ above and the strategies $\astrat_i^\epsilon$, we
now define a new family of strategies $\astrat_\epsilon$ for every $\epsilon >0$ for the game on $M_S$
from $c_0$. Given $\epsilon >0$, we let $d(\epsilon) = \alpha * \max_{i \in J}
d(i,\epsilon)$ (a number that is big enough for each projected component).

Playing from $c_0$ in $M_S$, the strategy $\astrat_\epsilon$ behaves as follows.
First it plays like strategy $\astrat$ in the corresponding game from $c_0'$
on $M_{S'}$. (Function $\lambda$ connects the corresponding configurations in the
two games.)
When the game in $M_{S'}$ reaches a configuration $r_i$ then there are two cases:
If the configuration in $M_S$ is $\ge s_i(d(\epsilon))$ then $\astrat_\epsilon$
henceforth plays like $\astrat_i^\epsilon$, which ensures to reach the target $X$
with probability $\ge 1-2\epsilon$.
Otherwise, the configuration in $M_S$ is still too small to switch to
$\astrat_i^\epsilon$.
In this case, $\astrat_\epsilon$ continues to play like $\astrat$ plays from
the previously visited smaller configuration in the uncolored part of $M_{S'}$ 
(see line 7 of the algorithm). This is possible, because the game is monotone
and larger configurations always benefit Player 1. So the game on $M_S$
continues with a configuration that is larger (at least on component $k(i)$) than
the corresponding game on $M_{S'}$, i.e., component $k(i)$ is pumped.
Since we know that $\astrat$ on $M_{S'}$ will almost surely visit $X_f'$ or $R$,
we obtain that $\astrat_\epsilon$ on $M_S$ will almost surely eventually visit
$X$ or some configuration $\ge s_i(d(\epsilon))$ for $i \in J$ (and from there
achieve to reach the target with probability $\ge 1-2\epsilon$).
Since every weighted average of probabilities $\ge 1-2\epsilon$ is still
$\ge 1-2\epsilon$, we obtain $\Prob(M_S, c_0, \astrat_\epsilon, \Interp{\eventually X}) \ge
1-2\epsilon$ and thus $\Prob^+(M_S, c_0, \Interp{\eventually X}) = 1$.
\qed
\end{proof}

\subsection{Proof of Theorem \ref{thm-limitsure-reach-1VAMDP}}

\begin{proof}
Let $S=\tuple{Q,Q_1,Q_P,T,\tau}$ be 1-VASS-MDP of dimension $d > 0$ with initial
configuration $c_0 = \tuple{q_0,\vec{v}}$ and $X\subseteq Q$ a set of target states.
We show decidability of $\mathbb{P}^+(M_S, c_0, \Interp{\eventually X}) = 1$ by
induction on $d$.
\paragraph{Base case}
$d=0$. If $S$ has $0$ counters then $M_S$ is a finite-state MDP and thus limit
sure reachability coincides with almost sure reachability, 
which is decidable. 

\paragraph{Inductive step.}
We apply Algorithm~\ref{alg:limitsure}, which terminates by
Lemma~\ref{lem:alg-limitsure-termination}, and obtain 
a new instance of the 1-VASS-MDP limit sure reachability problem of dimension
$d-1$: $S'=\tuple{Q',Q_1',Q_P',T',\tau'}$ with initial configuration 
$c_0' = \tuple{q_0',\vec{0}}$ and set of target states $X' \subseteq Q'$.
By Lemma~\ref{lem:limitsure_oldtonew} and Lemma~\ref{lem:limitsure_newtoold},
we have $\mathbb{P}^+(M_S, c_0, \Interp{\eventually X}) = 1 
\ \Leftrightarrow\ \mathbb{P}^+(M_{S'}, c_0', \Interp{\eventually X'}) = 1$.
By induction hypothesis, $\mathbb{P}^+(M_{S'}, c_0', \Interp{\eventually X'}) = 1$
is decidable and the result follows.
\qed
\end{proof}


\end{document}